\documentclass[10pt]{article}

\usepackage{xcolor}
\usepackage{geometry}
\usepackage{changepage}
\usepackage{bbm}
\usepackage[hidelinks]{hyperref}
\usepackage{amsmath,amsfonts,amssymb,amsthm, bm} 
\usepackage{graphicx} 
\usepackage{subcaption} 
\usepackage{adjustbox} 
\usepackage{float} 
\usepackage{algorithmicx} 
\usepackage{tabularx, comment} 
\usepackage{subcaption}

\newtheorem{theorem}{Theorem}

\newtheorem{Lemma}{Lemma}

\floatname{algorithm}{Algorithm}
\usepackage{algorithm}
\usepackage{algpseudocode}

\title{Exploring the Limitations of Structured Orthogonal Dictionary Learning}

\author{
  \begin{tabular}{cc}
    Anirudh Dash & Aditya Siripuram \\
    \small Department of Electrical Engineering & \small Department of Electrical Engineering \\
    \small Indian Institute of Technology, Hyderabad & \small Indian Institute of Technology, Hyderabad \\
    \small \texttt{ee21btech11002@iith.ac.in} & \small \texttt{staditya@ee.iith.ac.in}
  \end{tabular}
}
\date{}
\newcommand{\keywords}[1]{\textbf{Keywords:} #1}

\begin{document}
\maketitle

\begin{abstract}
    This work is motivated by recent applications of structured dictionary learning, in particular when the dictionary is assumed to be the product of a few Householder atoms. We investigate the following two problems: 1) How do we approximate an orthogonal matrix $\mathbf{V}$ with a product of a specified number of Householder matrices, and 2) How many samples are required to learn a structured (Householder) dictionary from data? For 1) we discuss an algorithm that decomposes $\mathbf{V}$ as a product of a specified number of Householder matrices. We see that the algorithm outputs the decomposition when it exists, and give bounds on the approximation error of the algorithm when such a decomposition does not exist. For 2) given data $\mathbf{Y}=\mathbf{HX}$, we show that when assuming a binary coefficient matrix $\mathbf{X}$, the structured (Householder) dictionary learning problem can be solved with just $2$ samples (columns) in $\mathbf{Y}$.
\end{abstract}


\section{Introduction and Related work}
\label{sec:intro}

The dictionary learning problem is as follows: given a data matrix \(\mathbf{Y} \in \mathbb{R}^{n \times p}\), the objective is to find a dictionary of atoms \( \mathbf{D} \in \mathbb{R}^{n \times m}\) and a coefficient matrix \(\mathbf{X} \in \mathbb{R}^{m \times p}\) such that \(\mathbf{Y}=\mathbf{DX}\). Generally, $m$ is chosen to be greater than $n$, and the redundancy is exploited to recover the dictionary. This problem has been well studied in literature (see \cite{olshausen1997sparse, engan1999method, aharon2006k, mairal2009online, sun2015complete} and this survey paper \cite{du2023matrix}). The orthogonal dictionary learning problem imposes the additional constraint: $m=n$. Several techniques have been developed to solve this problem, such as gradient descent \cite{gilboa2019efficient}, alternate minimization \cite{lesage2005learning, bao2013fast, liang2022simple}, $\ell^1$ minimization-based methods, and even using the $\ell^4$ norm to obtain theoretical guarantees for recovery of the orthogonal dictionary  \cite{zhai2020complete}. Online techniques have also been explored for orthogonal dictionary learning \cite{yeganli2014improved}. Recently, there has been an emphasis on fast dictionary learning, achieved by imposing additional structure on the dictionary by employing Householder reflections and Givens rotations \cite{rusu2016fast, rusu2017learning}.
For example, work in  \cite{rusu2016fast} constructs the dictionary as a product of $O(\text{log}n)$ Householder matrices (Recall that a householder matrix is given by $\mathbf{H} = \mathbf{I}-2\mathbf{uu}^\textsf{T}$ for a unit norm vector $\mathbf{u}$). Prior work by the authors in \cite{dash2024fast} gives a bound on the number of columns in $\mathbf{Y}$ that is required to recover a Householder dictionary $\mathbf{H}$ from $\mathbf{Y} = \mathbf{H X}$ under a statistical model on the coefficient matrix $\mathbf{X}$. Motivated by these developments, we hope to investigate the following questions in this work

\begin{enumerate}
    \item \emph{Richness of the product of a few Householder matrices:} How well can an orthogonal matrix be approximated by a product of a few Householder matrices? Recall that every $n \times n$ orthogonal matrix can be written as a product of $n$ Householder matrices \cite{uhlig2001constructive}. However, we can attempt to approximate an orthogonal matrix with the product of a few (preferably $O(\log n)$) Householder matrices. Such an approximation may have advantages in both computation speed and memory requirements to store the matrix.

    \item \emph{Lower bounds on sample complexity:} Similar to \cite{rusu2016fast, dash2024fast} Consider the dictionary learning problem as $\mathbf{Y}=\mathbf{HX}$ discussed earlier, where $\mathbf{H}$ is a Householder matrix. Are there any lower bounds on the sample complexity to recover $\mathbf{H}$ and $\mathbf{X}$ from $\mathbf{Y}$?
\end{enumerate}

As a first step towards solving these problems, we discuss 1) An algorithm that provably verifies if an input orthogonal matrix is a product of a certain number of Householder matrices and outputs the corresponding factors; and 2) A proof that when the coefficient matrix $\mathbf{X}$ is known to be binary, only two columns in the data matrix $\mathbf{Y}$ are sufficient to reconstruct $\mathbf{H}$.

The algorithm presented towards 1) above has a different construction of the Householder factors compared to prior work in \cite{uhlig2001constructive}, which enables giving error bounds on the performance of the algorithm at each iteration.

\section{Approximating with a few Householders}
\label{sec:not}
\subsection{Problem statement and discussion}
Given a target real, orthogonal matrix $\mathbf{V} \in \mathbb{R}^{n\times n}$ with $\mathbf{V}^\textsf{T}\mathbf{V}=\mathbf{I}$ and natural number $m$, we wish to determine real Householder matrices $\mathbf{H}_1,\mathbf{H}_2,\ldots,\mathbf{H}_m \in \mathbb{R}^{n \times n}$ that minimize \(\lVert \mathbf{V}-\mathbf{H}_1\mathbf{H}_2\ldots \mathbf{H}_m \rVert_F\). Here $\lVert \mathbf{A}\rVert_F = \sqrt{\text{tr}(\mathbf{A}^\textsf{T}\mathbf{A})}$ denotes the Frobenius norm of a matrix $\mathbf{A}$. We let $\mathcal{H}_m = \{\mathbf{A}\in \mathbb{R}^{n \times n}: \mathbf{A}=\mathbf{H}_1\mathbf{H}_2\ldots \mathbf{H}_k \text{ for some }k\leq m\}$ be the set of $n \times n$ real matrices that can be written as the product of some $k \leq m$ Householder matrices. We note that $\mathcal{H}_n$ includes all $n\times n$ orthogonal matrices.

In terms of computational complexity, $\mathbf{H}_i\mathbf{x}$ can be computed in $O(n)$ arithmetic operations for every $\mathbf{x} \in \mathbb{R}^n$, and in general for any $\mathbf{V}\in \mathcal{H}_m$, the matrix action $\mathbf{V}\mathbf{x}$ can be computed with $O(mn)$ arithmetic operations; if the decomposition $\mathbf{V}=\mathbf{H}_1\mathbf{H}_2\ldots \mathbf{H}_m$ is known. Further, these savings can be achieved by storing the underlying reflectors corresponding to the Householder $\mathbf{H}_i$; instead of storing the matrix $\mathbf{V}$ (the storage requirement thus comes down from $O(n^2)$ to $O(mn)$). This motivates the following problems 1) approximate arbitrary matrices $\mathbf{A}$ with those in $\mathcal{H}_m$ for $m\ll n$, and 2) find the decomposition $\mathbf{V} \approx \mathbf{H}_1\mathbf{H}_2 \ldots \mathbf{H}_k$.

We discuss an $O(n^3m)$ algorithm that, with input $\mathbf{V} \in \mathbb{R}^{n\times n}$ and $m$, finds an approximation to $\mathbf{V}$ in $\mathcal{H}_m$ and identifies the corresponding Householder matrices $\mathbf{H}_i$ such that $\mathbf{V} \approx \prod \mathbf{H}_i$. We show that the proposed algorithm correctly identifies if a given matrix
$\mathbf{V}\in \mathcal{H}_m$, and finds corresponding Householder factors $\mathbf{H}_i$. For an arbitrary $\mathbf{V}$, our algorithm can be seen as a greedy approach to finding the projection of $\mathbf{V}$ on $\mathcal{H}_m$. We also obtain error bounds on the approximation.

Note also that the Householder decomposition as above is not necessarily unique (example below). The classic Householder QR decomposition \cite{golub2013matrix} obtains the decomposition $\mathbf{A}=\mathbf{H}_1\mathbf{H}_2 \ldots \mathbf{H}_k\mathbf{R}$ for an arbitrary matrix $\mathbf{A}$ (here $\mathbf{R}$ is an upper triangular matrix). With input as an orthogonal matrix $\mathbf{V}$, the Householder QR decomposition can be used to decompose $\mathbf{V}$ into a product of Householder matrices; this means that any orthogonal matrix $\mathbf{V}$ can be expressed as $\mathbf{V}=\mathbf{H}_1\mathbf{H}_2\ldots \mathbf{H}_n \mathbf{R}$ for an orthogonal diagonal matrix $\mathbf{R}$ (i.e. $\mathbf{R}$ has diagonal entries $\pm 1$).

However, since the Householder QR algorithm operates column-wise, it typically returns $n$ Householder factors, even when the input matrix $\mathbf{V}\in \mathcal{H}_m$ has much fewer ($m<n$) Householder factors. Consider an example below with the input matrix $\mathbf{V}\in \mathbb{R}^{3\times 3}$. In particular, $\mathbf{V} \in \mathcal{H}_1$ is a Householder reflector. The Householder QR decomposition decomposes $\mathbf{V}$ into a product of $3$ Householder matrices and an additional upper triangular matrix as follows:

\[
    \mathbf{V}=
    \begin{bmatrix}
        1/9  & -4/9 & -8/9 \\
        -4/9 & 7/9  & -4/9 \\
        -8/9 & -4/9 & 1/9
    \end{bmatrix} = \mathbf{H}_1\mathbf{H}_2\mathbf{H}_3 \mathbf{R}, \text{ where }
\]
\begin{align*}
    \mathbf{H}_1 & = \begin{bmatrix}
                         -1/9 & 4/9    & 8/9    \\
                         4/9  & 37/45  & -16/45 \\
                         8/9  & -16/45 & 13/45
                     \end{bmatrix}, \quad
    \mathbf{H}_2 & = \begin{bmatrix}
                         1 & 0    & 0   \\
                         0 & -3/5 & 4/5 \\
                         0 & 4/5  & 3/5
                     \end{bmatrix}, \quad
    \mathbf{H}_3 & = \begin{bmatrix}
                         1 & 0 & 0  \\
                         0 & 1 & 0  \\
                         0 & 0 & -1
                     \end{bmatrix}, \quad
    \mathbf{R}   & = \begin{bmatrix}
                         -1 & 0  & 0 \\
                         0  & -1 & 0 \\
                         0  & 0  & 1
                     \end{bmatrix}
\end{align*}

The proposed algorithm (Algorithm \ref{alg:ortho_prod_of_Householders}), on the other hand, operates in the eigenspace (instead of the column space) and will correctly identify if $\mathbf{V} \in \mathcal{H}_1$.

We hope that such approaches may be useful as a pre-computation step to decompose matrices $\mathbf{V}$ into Householder factors so that $\mathbf{V}\mathbf{x}$ can be approximated efficiently.

We discuss the main results in the next section. Some notations before we proceed: we denote by $\mathbf{u}_i$ the unit norm vectors corresponding to the Householder reflectors $\mathbf{H}_i$, i.e. $\mathbf{H}_i=\mathbf{I}-2\mathbf{u}_i\mathbf{u}_i^\textsf{T}$. We denote by $\mathbf{V}_{\text{sym}}=(\mathbf{V}+\mathbf{V}^\textsf{T})/2$ the symmetric part of an orthogonal matrix $\mathbf{V}$. We also denote but $\mathbf{I}_n$ (or simply $\mathbf{I}$ when $n$ is clear from the context) the $n\times n$ identity; for symmetric matrices $\mathbf{A}$, we denote by $\lambda_{\text{min}}(\mathbf{A})$ the smallest eigenvalue of $\mathbf{A}$. We use $\text{tr}(\mathbf{M})$ to represent the trace of the matrix $\mathbf{M}$, $\text{det}(\mathbf{M})$ to denote its determinant and $\text{mult}(\mathbf{M},\lambda)$ to denote the (geometric) multiplicity of its eigenvalue $\lambda$.

\subsection{Main Results}

Before we state the result, we make the following simple observations. Let $\mathcal{E}_\mathbf{V}^1$ be the eigenspace of the matrix $\mathbf{V}$ corresponding to the eigenvalue $1$, i.e. \(\mathcal{E}_\mathbf{V}^1= \{\mathbf{x}: \mathbf{V}\mathbf{x}=\mathbf{x}\} \). We see that for Householder matrices $\mathbf{H}_i = \mathbf{I} - 2\mathbf{u}_i\mathbf{u}^{\textsf{T}}_i$, the space $\mathcal{E}_{\mathbf{H}_i}^1$ is the $n-1$ dimensional subspace orthogonal to $\mathbf{u}_i$. Note also that if $\mathbf{V}= \mathbf{H}_1\mathbf{H}_2\ldots \mathbf{H}_m,$ then
\[
    \cap \mathcal{E}_{\mathbf{H}_i}^1 \subseteq \mathcal{E}_\mathbf{V}^1.
\]
Thus, the eigenspace of $\mathbf{V}$ corresponding to eigenvalue $1$ includes the subspace $\cap \mathcal{E}_{\mathbf{H}_i}^1$: the latter subspace is orthogonal to all $\mathbf{u}_1, \mathbf{u}_2,\ldots, \mathbf{u}_m$ and hence is of dimension at least $n-m$. Thus if $\mathbf{V}= \mathbf{H}_1\mathbf{H}_2\ldots \mathbf{H}_m \in \mathcal{H}_m$, then $\mathcal{E}_\mathbf{V}^1$ is at least $n-m$ dimensional; this forms a necessary condition for $\mathbf{V} \in \mathcal{H}_m$. This condition is known to be sufficient as well \cite{uhlig2001constructive}.

If $\mathbf{V}$ is known to be symmetric, we can find the decomposition easily: to see this, note that a symmetric orthogonal $\mathbf{V}$ has a full set of orthogonal eigenvectors with eigenvalues $\pm 1$, and so has a decomposition of the form
\[
    \mathbf{V}= \underbrace{\mathbf{v}_1\mathbf{v}_1^\textsf{T}+ \mathbf{v}_2\mathbf{v}_2^\textsf{T}+ \ldots + \mathbf{v}_k\mathbf{v}_k^\textsf{T}}_{\geq n-m \text{ terms}} -  \mathbf{v}_{k+1}\mathbf{v}_{k+1}^\textsf{T} -\ldots-\mathbf{v}_n\mathbf{v}_n^\textsf{T}.
\]
Using $\sum \mathbf{v}_i\mathbf{v}_i^\textsf{T} = \mathbf{I}$, we may rewrite the above as
\begin{equation}
    \mathbf{V} = \mathbf{I} - 2 \sum_{i=k+1}^n\mathbf{v}_i\mathbf{v}_i^\textsf{T} = \prod_{i=k+1}^n\left(\mathbf{I}-2\mathbf{v}_i\mathbf{v}_i^\textsf{T} \right), \label{eq:symmetric_Householder_decomposition}
\end{equation}
the last equality from sum to product follows from the orthogonality of the eigenvectors $\mathbf{v}_i$. Thus, \eqref{eq:symmetric_Householder_decomposition} gives a straightforward way to find the Householder decomposition for symmetric $\mathbf{V} \in \mathcal{H}_m$: we simply find the eigenvectors $\mathbf{v}_{k+1},\mathbf{v}_{k+2}, \ldots, \mathbf{v}_{n}$ corresponding to eigenvalue $-1$, and generate the Householder factors $\mathbf{H}_i$ as $\mathbf{H}_i = \mathbf{I}- 2\mathbf{v}_i\mathbf{v}_i^\textsf{T}$. Thus for symmetric orthogonal matrices, finding the Householder decomposition is equivalent to the eigenvalue decomposition. The result below generalizes this approach to arbitrary (non-symmetric) real orthogonal matrices with the following goals: Given $\mathbf{V}$ and $m$: 1) identify if $\mathbf{V} \in \mathcal{H}_m$, and 2) Find $\mathbf{H}_1, \mathbf{H}_2,\ldots, \mathbf{H}_m$ such that $\mathbf{V} \approx \mathbf{H}_1 \mathbf{H}_2\ldots \mathbf{H}_m$, and 3) give a bound on the error in this approximation.

\subsection{Algorithm}
Consider the case when $m=1$, i.e. approximating an input matrix $\mathbf{V}$ with a single Householder matrix $\mathbf{H}$. We have the target squared error objective $\lVert \mathbf{V}-\mathbf{H} \rVert_F^2$, with $\mathbf{H} = \mathbf{I} - 2 \mathbf{u} \mathbf{u}^T$:
\begin{align}
    \min_{\mathbf{H} \in \mathcal{H}_1}\lVert \mathbf{V}-\mathbf{H}\rVert_F^2 & = \min_{\mathbf{H} \in \mathcal{H}_1}\left(2n - 2 \text{tr}(\mathbf{H}^\textsf{T}\mathbf{V})\right)\nonumber \\ &= \min_{\lVert \mathbf{u} \rVert_2 = 1}\left(2n - 2 \text{tr}(\mathbf{V}) + 2\mathbf{u}^\textsf{T}\mathbf{Vu}\right)\nonumber\\
                                                                              & =2n - 2 \text{tr}(\mathbf{V})+ 2 \lambda_\text{min}(\mathbf{V}_{\text{sym}}), \label{eq:approx-in-H1}
\end{align}
where in the last step we have used $\min (\mathbf{u}^\textsf{T} \mathbf{V u}) = \min (\mathbf{u}^\textsf{T} \mathbf{V}_{\text{sym}}\mathbf{u}) = \lambda_{\text{min}}(\mathbf{V}_{\text{sym}})$. Thus, the projection of $\mathbf{V}$ in $\mathcal{H}_1$ is obtained as $\mathbf{I}- 2 \mathbf{u u}^\textsf{T}$, where $\mathbf{u}$ is a unit eigenvector corresponding to the smallest eigenvalue of the symmetric part $\mathbf{V}_{\text{sym}}$.

For approximation in $\mathcal{H}_m$, for $m>1$, The error $\lVert \mathbf{V} - \mathbf{H} \rVert_F^2$ may not reduce as favorably as in the $m=1$ case. Hence we adopt the following strategy to find a suboptimal solution: first, approximate the input $\mathbf{V}$ in $\mathcal{H}_1$ (as in observation \eqref{eq:approx-in-H1} above). Let $\mathbf{H}_1$ be the approximation obtained. We then construct $\mathbf{V}_1=\mathbf{H}_1^\textsf{T}\mathbf{V} = \mathbf{H}_1 \mathbf{V}$ and approximate $\mathbf{V}_1$ in $\mathcal{H}_1$ to obtain $\mathbf{H}_2$, and repeat the process for $\mathbf{V}_2 = \mathbf{H}_2\mathbf{V}_1$. This algorithm is summarized in Algorithm \ref{alg:ortho_prod_of_Householders}. Theorem \ref{thm:correctness} below gives performance guarantees on this suboptimal algorithm.
Note that the expression for an arbitrary $\mathbf{V}$ in $\mathcal{H}_m$ is of the form:
\begin{align*}
    \mathbf{V} & = \mathbf{I} - 2\sum_{i=1} \mathbf{u}_i \mathbf{u}_i^\textsf{T}
    + 4\sum_{\substack{i,j                                                       \\ i<j}} k_{ij} \mathbf{u}_i \mathbf{u}_j^\textsf{T}
    - 8\sum_{\substack{i,j,l                                                     \\ i<j; j<l}} k_{ij} k_{jl} \mathbf{u}_i \mathbf{u}_l^\textsf{T}
               & \cdots + (-2)^m\sum_{\substack{i,j,l \cdots                     \\ i<j; j<l; \cdots}} (k_{ij} k_{jl} \cdots )\mathbf{u}_i \mathbf{u}_m^\textsf{T}
\end{align*}
where $\mathbf{u}_1,\mathbf{u}_2, \ldots$ are unit norm vectors and $k_{ij} = \mathbf{u}_i^\textsf{T}\mathbf{u}_j$ are the pairwise inner products.

\begin{theorem}
    \label{thm:correctness}
    If Algorithm \ref{alg:ortho_prod_of_Householders} terminates at step $m$ then $\mathbf{V} \in \mathcal{H}_m$. Furthermore, this is the smallest $m$ such that $\mathbf{V} \in \mathcal{H}_m$.

\end{theorem}
Note that we can modify Algorithm \ref{alg:ortho_prod_of_Householders} to find an approximation of an arbitrary $\mathbf{V}$ in $\mathcal{H}_m$ by truncating it to $m$ steps.
\begin{theorem}
    \label{thm:errorbound}
    For an arbitrary input $\mathbf{V} \in \mathbb{R}^{n \times n}$, the Householder decomposition $\hat{\mathbf{V}}_m$ obtained from Algorithm \ref{alg:ortho_prod_of_Householders} (truncated to $m$ steps) satisfies the error bound

    \begin{align*}
        \lVert \mathbf{V} - \hat{\mathbf{V}}_m \rVert_F \leq
        \sqrt{2\left(n - \operatorname{tr}(\mathbf{V}) - 2\lfloor m/2 \rfloor
            + \sum_{i=1}^{m} \lambda_i \right)}
    \end{align*}

    where the $\lambda_1\leq \lambda_2 \leq \ldots$ are the eigenvalues of $\mathbf{V}_{\text{sym}}$.
\end{theorem}

Note that the Householder reflectors obtained by the algorithm are not necessarily orthogonal (as in the case when the input $\mathbf{V}$ is symmetric). The construction in \cite{uhlig2001constructive} does not use the eigenvectors of the matrix and thus does not come with approximation bounds.

\subsection{Proof}
\begin{Lemma}
    \label{lem:ortho-basics}
    We recall the following known facts about real orthogonal matrices. We use these observations for the proof of Theorem \ref{thm:correctness} and Theorem \ref{thm:errorbound}.
    \begin{enumerate}
        \item Any orthogonal matrix $\mathbf{V} \in \mathbb{R}^{n \times n}$ is normal, and hence has a decomposition of the from $\mathbf{V} = \mathbf{U D U}^\star$ where $\mathbf{U}\in \mathbb{C}^{n \times n}$ is a unitary matrix, and $\mathbf{D} \in \mathbb{C}^{n\times n}$ is a diagonal matrix with diagonal entries complex units. In particular, eigenvectors corresponding to different eigenvalues of $\mathbf{V}$ are orthogonal. 
        \item If $\lambda$ is an eigenvalue of $\mathbf{V}$, so is $\bar{\lambda}$. If $\mathbf{w}$ is an eigenvector corresponding to $\lambda$, then $\bar{\mathbf{w}}$ is an eigenvector corresponding to $\bar{\lambda}$. Also, since the complex eigenvalues appear in conjugate pairs, $\det (\mathbf{V}) = \prod \lambda_i=-1$ if and only if $-1$ is an eigenvalue of $\mathbf{V}$ with odd multiplicity.
        \item The Eigenvectors of $\mathbf{V}_\text{sym}$ are the diagonal entries of $\text{Re}(\lambda)$ (counted with multiplicities). If $\mathbf{V w} = \lambda \mathbf{w}$ then $\text{Re}(\lambda)$ is an eigenvalue of $\mathbf{V}_\text{sym}$ with corresponding eigenvectors $\mathbf{w}+\bar{\mathbf{w}}$ and $(\mathbf{w}-\bar{\mathbf{w}})/i$. Likewise, every eigenvector of $\mathbf{V}_{\text{sym}}$ is associated with eigenvectors $\mathbf{z}$ and $\bar{\mathbf{z}}$ of $\mathbf{V}.$
        \item If $\mathbf{w}$ is an eigenvector of $\mathbf{V}$ with eigenvalue $\lambda$, then $\mathbf{w}$ is an eigenvector of $\mathbf{V}^\textsf{T}$ with eigenvalue $\bar{\lambda}$. If $\lambda=\pm 1$, then any eigenvector of $\mathbf{V}_\text{sym}$ with eigenvalue $\lambda$ is also an eigenvector of $\mathbf{V}$ with eigenvalue $\lambda$.
    \end{enumerate}
\end{Lemma}

\begin{proof}
    $3.$ and the first statement of $4.$ can be proved by noting that $\mathbf{V}^\textsf{T}\mathbf{V}=\mathbf{V}\mathbf{V}^\textsf{T}=\mathbf{I}$ and the conjugate transpose of a real matrix is real.
    Using the triangle inequality for $\mathbf{V}^{\textsf{T}}\mathbf{w}=-2\mathbf{w}-\mathbf{V}\mathbf{w}$ suffices to prove the second statement of $4.$

\end{proof}
We note, in particular, the following:

\begin{Lemma}
    Consider the vector $\mathbf{u}$ in the algorithm at step $k$ (the eigenvector corresponding to eigenvalue $\lambda_{\min}$ of $(\mathbf{V}_k)_{\text{sym}}$). Let $\{\mathbf{z},\bar{\mathbf{z}}\}$ be the eigenvectors of $\mathbf{V}_k$ associated to $\mathbf{u}$. Every eigenvector $\mathbf{w}$ of $\mathbf{V}_k$ other than $\{\mathbf{u},\bar{\mathbf{u}} \}$ is an eigenvector of $\mathbf{V}_{k+1}$. Further, if $\lambda_{\min}=-1$, the vector $\mathbf{u}$ picked in step $k$ of the algorithm satisfies $\mathbf{V}_{k+1}\mathbf{u}=(\mathbf{V}_{k+1})_\text{sym}\mathbf{u} = \mathbf{u}$.
    \label{lem:spectral2}
\end{Lemma}

\begin{proof}
    Note that $\mathbf{V}_{k+1}\mathbf{w}=\mathbf{HV}_k\mathbf{w}=\mathbf{Vw}-2\mathbf{uu}^\textsf{T}\mathbf{Vw}.$ Since $\mathbf{u}$ is in the span of $\mathbf{z}, \bar{\mathbf{z}}$, it follows that $\mathbf{u}^\textsf{T}\mathbf{w} = 0$ (since $\mathbf{z}, \bar{\mathbf{z}}$ are both orthogonal to $\mathbf{w}$ on $\mathbb{C}^n$). So we get $\mathbf{V}_{k+1} \mathbf{w} = \mathbf{V}_k \mathbf{w}$, completing the first part. The second part follows from a similar calculation.
\end{proof}

To prove Theorem \ref{thm:correctness}, we show that the eigenspace $\mathcal{E}^1_{\mathbf{V}_k}$ corresponding to the eigenvalue $1$ increases in dimension by $1$ for each iteration.
\begin{Lemma}The trace and eigenspace for eigenvalue $1$ increase at each iteration:
    $\text{tr}(\mathbf{V}_{k+1}) - \text{tr}(\mathbf{V}_k) = -2\lambda_{\min}(\mathbf{V}_k)_{\text{sym}}$, and
    $\text{dim}(\mathcal{E}^1_{\mathbf{V}_{k+1}}) - \text{dim}(\mathcal{E}^1_{\mathbf{V}_{k}}) = 1$.
    \label{lem:trace3}
\end{Lemma}
For the trace, we see $\text{tr}(\mathbf{V}_{k+1})=\text{tr}(\mathbf{H}\mathbf{V}_{k})=\text{tr}(\mathbf{V}_k)-2\text{tr}(\mathbf{u}^\textsf{T}\mathbf{V}_k\mathbf{u})$, giving the required expression.

For the eigenspace, suppose that the input $\mathbf{V} \in \mathcal{H}_p$ is a product of $p$ Householder matrices for $p\leq n$. Recall that every $n \times n$ orthogonal matrix is a product of at most $n$ Householder matrices \cite{uhlig2001constructive}. Here, we distinguish the following two cases
\begin{enumerate}
    \item \emph{Case 1:} $p+k$ is odd: In this case, $\mathbf{V}_k =\mathbf{H}_{k}\mathbf{H}_{k-1}\ldots \mathbf{H}_1 \mathbf{V} $ is product of $p+k$ Householder matrices, and so $\det (\mathbf{V}_k) = (-1)^{p+k}=-1$. It follows that $-1$ is an eigenvalue of $\mathbf{V}_k$ and $(\mathbf{V}_k)_\text{sym}$.  Thus, by the construction of the algorithm, the vector $\mathbf{u}$ picked at step $k$ satisfies $(\mathbf{V}_k)_{\text{sym}}\mathbf{u} = \mathbf{V}_k \mathbf{u} = -\mathbf{u}$. From Lemma \ref{lem:spectral2}, it follows that the $\mathcal{E}^1_{\mathbf{V}_{k+1}} = \mathcal{E}^1_{\mathbf{V}_k} + \cup_\alpha\{\alpha u\}$, and the Lemma follows.
    \item \emph{Case 2:} $p+k$ is even: If $\lambda_{\min}=\lambda_{\min}(\mathbf{V}_k)_{\text{sym}}=-1$, we make a similar argument to the above case. So suppose $\lambda_{\min}>-1$.  Note that $\mathbf{V}_{k+1}$ is a product of $p+k+1$ Householder matrices and hence has an eigenvalue of $-1$ with odd multiplicity. For complex unit $\lambda$, let $k_\lambda = \text{mult}(\mathbf{V}_{k+1}, \lambda) - \text{mult}(\mathbf{V}_k, \lambda)$. Note that by definition $\sum k_{\lambda} = 0$ (the sum is over eigenvalues of either $\mathbf{V}_k$ or $\mathbf{V}_{k+1}$). From Lemma \ref{lem:spectral2}, since all eigenvectors other than $\mathbf{u}$ and $\bar{\mathbf{u}}$ are retained, so $|k_\lambda| \leq 2$. Also, for any $\lambda$ other than those associated with $\mathbf{u}, \bar{\mathbf{u}}$, the multiplicity cannot decrease, so these $k_\lambda$ are non-negative; in particular $k_1\geq 0$ since $\lambda_{\min}\neq 1$. As observed previously, $k_{-1}$ is odd (positive); Furthermore, because $\mathbf{V}_{k+1}$ and $\mathbf{V}_k$ are real orthogonal matrices, the complex eigenvalues occur in conjugate pairs, and so $k_\lambda = k_{\bar{\lambda}}$. 
          Hence, the only way these conditions hold is with $k_{-1}=1$ and $k_{1}=1,$ which means the eigenspace for eigenvalue $1$ increases in dimension by $1$.
\end{enumerate}

\subsubsection{Proof of Theorem \ref{thm:correctness} and Theorem \ref{thm:errorbound}}
Algorithm \ref{alg:ortho_prod_of_Householders} terminates when $\hat{\mathbf{V}}=\mathbf{V}$ or equivalently when $\mathbf{V}_{m+1}=\mathbf{I}$. In this case, by construction, $\hat{\mathbf{V}}$ is a product of $m$ Householder matrices. Now suppose by way of contradiction $\mathbf{V}\in \mathcal{H}_p$ for some $p<m$. Then as observed before, the eigenspace $\mathcal{E}^1_{\mathbf{V}}$ is at least $n-p$ dimensional. By Lemma \ref{lem:trace3} above, at iteration $p$, $\mathbf{V}_{p+1}$ would have an $n-$dimensional eigenspace; so the algorithm should have terminated before step $p$. \\

For Theorem \ref{thm:errorbound}, we have ${\text{min}} \ \lVert \mathbf{V} - \hat{\mathbf{V}} \rVert_F^2 ={\text{min}} \ 2[n - \text{tr}(\hat{\mathbf{V}}^\textsf{T}\mathbf{V})]=2[n - \text{tr}({\mathbf{V}_{k+1}})]$. Note that Lemma \ref{lem:trace3} gives a recursion on $\text{tr}(\mathbf{V}_{k+1})$ in terms of the smallest eigenvalue of $(\mathbf{V}_{k})_\text{sym}$. Similar to the proof of Lemma \ref{lem:trace3}, we note that
the smallest eigenvalue of $(\mathbf{V}_k)_{\text{sym}}$ is $-1$ for every other iteration (when $p+k$ as in Proof of Lemma \ref{lem:trace3} is odd). In any iteration, no new eigenvalues other than $\pm 1$ are introduced. Thus $\lambda_{\min}(\mathbf{V}_k)_{\text{sym}}$ is the $k/2^{\textsf{th}}$ smallest eigenvalue  of $\mathbf{V}_\text{sym}$.
\[
    \lVert \mathbf{V} - \hat{\mathbf{V}_m} \rVert_F^2 = \\
    2\left(n - \operatorname{tr}(\mathbf{V}) - 2\sum\lambda_{\min}(V_k)_{\text{sym}}\right),
\]

so the statement of the theorem follows (note that the eigenvalues appear in pairs, allowing us to rewrite as in the statement of the theorem).

\subsection{Illustration of Householder Recovery}
We illustrate the recovery of an arbitrary orthogonal matrix $\mathbf{V} \in \mathcal{H}_2$, i.e., we find matrcies $\mathbf{H}_3, \ \mathbf{H}_4$ such that
$\mathbf{V}=\mathbf{H}_1\mathbf{H}_2=\mathbf{H}_3 \mathbf{H}_4$.

\(\mathbf{V} = (\mathbf{I} - 2\mathbf{u}_1\mathbf{u}_1^\textsf{T})(\mathbf{I} - 2\mathbf{u}_2\mathbf{u}_2^\textsf{T}) = \mathbf{I} - 2\mathbf{u}_1\mathbf{u}_1^\textsf{T} - 2\mathbf{u}_2\mathbf{u}_2^\textsf{T} + 4 (\mathbf{u}_1^\textsf{T} \mathbf{u}_2) \mathbf{u}_1 \mathbf{u}_2^\textsf{T}\).
Let \(\mathbf{u}_1^\textsf{T} \mathbf{u}_2 = k\) for notational simplicity. Then,
\((\mathbf{V} + \mathbf{V}^\textsf{T})/2 = \mathbf{I} - 2\mathbf{u}_1\mathbf{u}_1^\textsf{T} - 2\mathbf{u}_2\mathbf{u}_2^\textsf{T} + 2k (\mathbf{u}_1\mathbf{u}_2^\textsf{T} + \mathbf{u}_2\mathbf{u}_1^\textsf{T})\).
Now,
\(\mathbf{V}_{\text{sym}} \mathbf{u}_1 = (\mathbf{I} - 2\mathbf{u}_1\mathbf{u}_1^\textsf{T} - 2\mathbf{u}_2\mathbf{u}_2^\textsf{T} + 2k (\mathbf{u}_1 \mathbf{u}_2^\textsf{T} + \mathbf{u}_2 \mathbf{u}_1^\textsf{T})) \mathbf{u}_1=-\mathbf{u}_1+2k^2\mathbf{u}_1\). Similarly, \(\mathbf{V}_{\text{sym}} \mathbf{u}_2 = (\mathbf{I} - 2\mathbf{u}_1\mathbf{u}_1^\textsf{T} - 2\mathbf{u}_2\mathbf{u}_2^\textsf{T} + 2k (\mathbf{u}_1 \mathbf{u}_2^\textsf{T} + \mathbf{u}_2 \mathbf{u}_1^\textsf{T})) \mathbf{u}_2=-\mathbf{u}_2+2k^2\mathbf{u}_2\). Thus, $\mathbf{u}_1$ and $\mathbf{u}_2$ are eigenvectors with the same eigenvalue $(-1+2k^2)$. Hence, any vector in the span of $\mathbf{u}_1$ and $\mathbf{u}_2$, i.e., $\alpha \mathbf{u}_1+ \beta \mathbf{u}_2$ is also an eigenvector. Imposing the unit norm constraint on the eigenvector, $\lVert \alpha \mathbf{u}_1+ \beta \mathbf{u}_2 \rVert^2 =1$. Thus, $\alpha^2+\beta^2 +2\alpha \beta k=1.$ Consider any vector (say, $\mathbf{w}$) orthogonal to both $\mathbf{u}_1$ and $\mathbf{u}_2$. Thus, $\mathbf{V_{\text{sym}}w}=\mathbf{w}$. All such vectors have eigenvalue $1$. Thus, the minimum eigenvalue is $(-1+2k^2)$, since $-1 \leq k \leq 1$. Set $\mathbf{H}_3=\mathbf{I}-2(\alpha \mathbf{u}_1+ \beta \mathbf{u}_2)(\alpha \mathbf{u}_1+ \beta \mathbf{u}_2)^\textsf{T}=\mathbf{I}-2(\alpha^2 \mathbf{u}_1\mathbf{u}_1^\textsf{T}+\alpha \beta (\mathbf{u}_1\mathbf{u}_2^\textsf{T}+\mathbf{u}_2\mathbf{u}_1^\textsf{T})+\beta^2 \mathbf{u}_1\mathbf{u}_1^\textsf{T})$. Now, \(\mathbf{H}_3^{-1} \mathbf{V} = \mathbf{H}_3 \mathbf{H}_1 \mathbf{H}_2 = \mathbf{I}
- 2\mathbf{u}_1\mathbf{u}_1^\textsf{T}(1 - \alpha^2 - 2\alpha\beta k)
- 2\mathbf{u}_2\mathbf{u}_2^\textsf{T}(1 - \beta^2 + 2\alpha\beta k + 4\beta^2 k^2)
- 2\mathbf{u}_1\mathbf{u}_2^\textsf{T}(-2k + 2k\alpha^2 - \alpha\beta + 4\alpha\beta k^2)
- 2\mathbf{u}_2\mathbf{u}_1^\textsf{T}(-\alpha \beta - 2 \beta^2 k)\). Note that $(-2k + 2k\alpha^2 - \alpha\beta + 4\alpha\beta k^2)
    =- \alpha\beta + 2k(\alpha^2+2 \alpha \beta k -1)= -\alpha \beta - 2 \beta^2 k$, by using the unit norm condition. Thus, the coefficients
of both $2\mathbf{u}_1\mathbf{u}_2^\textsf{T}$ and $ 2\mathbf{u}_2\mathbf{u}_1^\textsf{T}$
are equal to $(-\alpha\beta - 2\beta^2 k)$. \\

Finally, consider the vector $\mathbf{v}= -\beta \mathbf{u}_1+(\alpha + 2\beta k)\mathbf{u}_2$.
Let $\mathbf{H}_4=\mathbf{I}-2(-\beta \mathbf{u}_1+ (\alpha+2\beta k)\mathbf{u}_2)(-\beta \mathbf{u}_1+ (\alpha+2\beta k)\mathbf{u}_2)^\textsf{T}$.
It can be verified that $\mathbf{H}_3\mathbf{H}_1\mathbf{H}_2=\mathbf{H}_4$. All that is left is to analyze whether $-\beta \mathbf{u}_1+(\alpha + 2\beta k)\mathbf{u}_2$
is an eigenvector of $(\mathbf{H}_3\mathbf{H}_1\mathbf{H}_2)_{\text{sym}}=\mathbf{I}-2\mathbf{u}_1\mathbf{u}_1^\textsf{T} - 2\mathbf{u}_2\mathbf{u}_2^\textsf{T}
    -2(\alpha\beta + 2\beta^2 k)(\mathbf{u}_1\mathbf{u}_2^\textsf{T}+\mathbf{u}_2\mathbf{u}_1^\textsf{T})$. $(\mathbf{H}_3\mathbf{H}_1\mathbf{H}_2)_{\text{sym}}-\beta \mathbf{u}_1+(\alpha + 2\beta k)\mathbf{u}_2
    =-\beta \mathbf{u}_1+2\beta \mathbf{u}_1(1-\alpha^2-2\alpha\beta k)+2\beta k \mathbf{u}_2(1-\beta^2 +2 \alpha\beta k + 4\beta^2 k^2)
    -2 \beta \mathbf{u}_2 (\alpha \beta + 2\beta^2 k) - 2\beta k \mathbf{u}_1(\alpha \beta + 2 \beta^2 k) + \mathbf{u}_2(\alpha+2 \beta k)
    -2k\mathbf{u}_1(\alpha + 2 \beta k)(1-\alpha^2 - 2\alpha \beta k)-2\mathbf{u}_2(\alpha+2\beta k)(1-\beta^2 + 2\alpha\beta k + 4\beta^2 k^2)
    +2k\mathbf{u}_2(\alpha + 2\beta k)(\alpha\beta + 2 \beta^2 k)+2\mathbf{u}_1(\alpha+2\beta k)(\alpha \beta + 2\beta^2 k)$.
Clubbing the coefficents of $\mathbf{u}_1$ and $\mathbf{u}_2$, we get
$\mathbf{u}_1(-\beta + 2\beta(1-\alpha^2-2\alpha\beta k)-2\beta k(\alpha \beta + 2\beta^2 k)-2k(\alpha+2\beta k)(1-\alpha^2 - 2\alpha \beta k)+2(\alpha+2\beta k)(\alpha \beta + 2\beta^2 k))+
    \mathbf{u}_2(2\beta k(1-\beta^2 +2 \alpha\beta k + 4\beta^2 k^2)-2\beta(\alpha \beta + 2\beta^2 k)+(\alpha + 2\beta k)-2(\alpha+2\beta k)(1-\beta^2 + 2\alpha\beta k + 4\beta^2 k^2)+2k(\alpha + 2\beta k)(\alpha\beta + 2 \beta^2 k))$.
$=\mathbf{u}_1(-\beta + 2\beta^3-2\alpha\beta^2 k -4\beta^3 k^2-2\alpha\beta^2 k-4\beta^3 k^2+2\alpha^2\beta+4\alpha\beta^2 k+4\alpha\beta^2 k+8\beta^3 k^2)$
$+\mathbf{u}_2(2\beta k (\alpha+2\beta k)^2-2\beta^2(\alpha+2\beta k)+(\alpha+2\beta k)-2(\alpha+2\beta k)^3)$
$=\mathbf{u}_1(-\beta + 2\beta^3+4\alpha \beta^2 k+ 2\alpha^2 \beta)+
    \mathbf{u}_2(\alpha+2\beta k)(4\beta k (\alpha+2\beta k)-2\beta^2+1-2(\alpha+2\beta k)^2)=
    -\beta \mathbf{u}_1 (1-2\beta^2 -2\alpha^2-4\alpha \beta k)+ (\alpha+2\beta k)\mathbf{u}_2(1-2\beta^2-2\alpha^2-4\alpha \beta k)$
$=-\beta \mathbf{u}_1+ (\alpha+2\beta k)\mathbf{u}_2(-1)$ (using the unit norm condition). \\

Therefore, the vector $\mathbf{v}= -\beta \mathbf{u}_1+(\alpha + 2\beta k)\mathbf{u}_2$ is an eigenvector of $(\mathbf{H}_3\mathbf{H}_1\mathbf{H}_2)_{\text{sym}}$ with eigenvalue $-1$, which is the minimum eigenvalue.
Moreover, the eigenvector corresponding to this eigenvalue, which turns out to be the Householder vector for $\mathbf{H}_4$, is unique up to sign.
Thus, note the following: $\mathbf{V_{\text{sym}}}$ had the eigenvalue $(-1+2k^2)$ with multiplicity $2$ and the eigenvalue $1$ with multiplicity $(n-2)$.
If we consider the symmetric part of $\mathbf{H}_3\mathbf{H}_1\mathbf{H}_2,$ it tur
ns out to be equal to
$\mathbf{H}_3\mathbf{H}_1\mathbf{H}_2.$ Once again, any vector orthogonal to both $\mathbf{u}_1,\mathbf{u}_2$ is
still an eigenvector with eigenvalue $1.$ However, unlike the symmetric part of $\mathbf{H}_1\mathbf{H}_2$,
whose minimum eigenvalue had multiplicity $2,$ the minimum eigenvalue in this case ($-1$), which corresponds to the
specific eigenvector $-\beta \mathbf{u}_1+(\alpha + 2\beta k)\mathbf{u}_2$ has multiplicity $1.$
Since the matrix under consideration is symmetric, the space spanned by the eigenvectors has dimension $n$.
Thus, there exists another eigenvector in the plane of $\mathbf{u}_1, \mathbf{u}_2$, orthogonal to the above Householder vector.
It can further be verified that this eigenvector has eigenvalue $1$. Thus, there are infinite solutions to the equation
$\mathbf{H}_1\mathbf{H}_2=\mathbf{H}_3\mathbf{H}_4$, since $\alpha$ and $\beta$ can be chosen arbitrarily, as long as the unit norm condition is satisfied.

\begin{figure}[htbp]
    \centering
    \begin{subfigure}[t]{0.45\linewidth}
        \includegraphics[width=\linewidth]{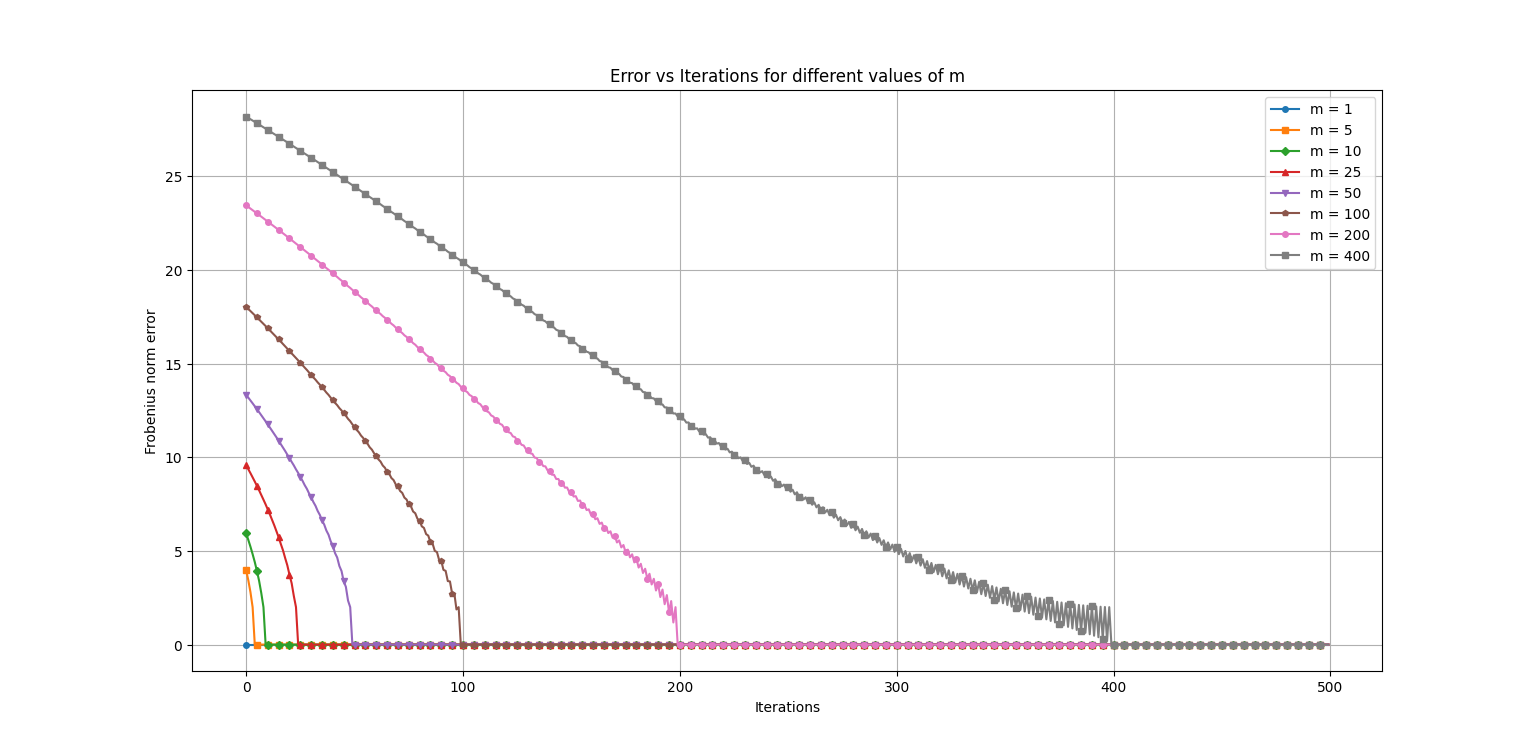}
        \caption{The Householder vectors are generated from an arbitrary Gaussian distribution, followed by normalization.}
    \end{subfigure}
    \hfill
    \begin{subfigure}[t]{0.45\linewidth}
        \includegraphics[width=\linewidth]{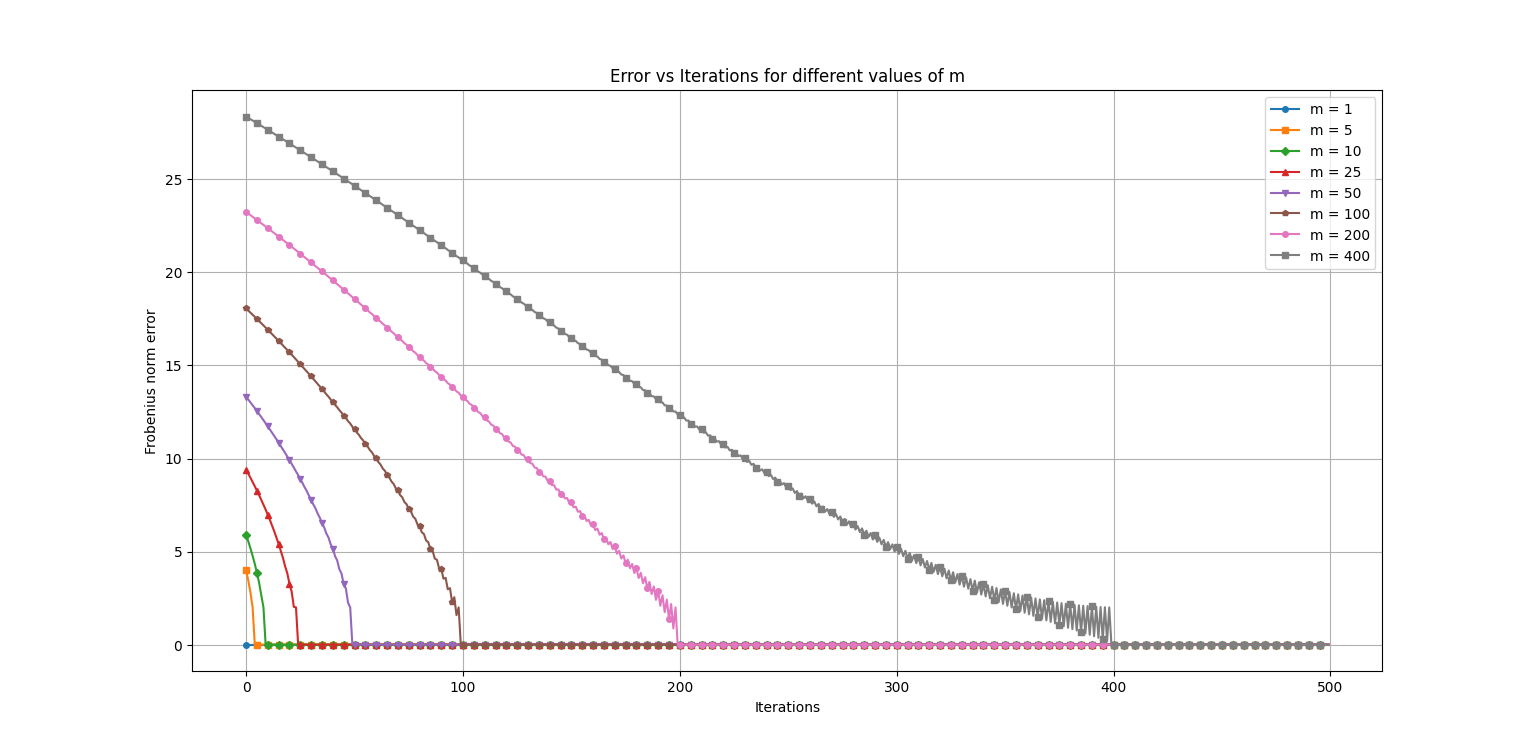}
        \caption{A small fraction of entries $k=0.02$ is chosen randomly. These entries are drawn from an arbitrary Gaussian distribution, while all other entries are set to $0$. The vector is then normalized.}
    \end{subfigure}

    \vspace{0.5cm}

    \begin{subfigure}[t]{0.45\linewidth}
        \includegraphics[width=\linewidth]{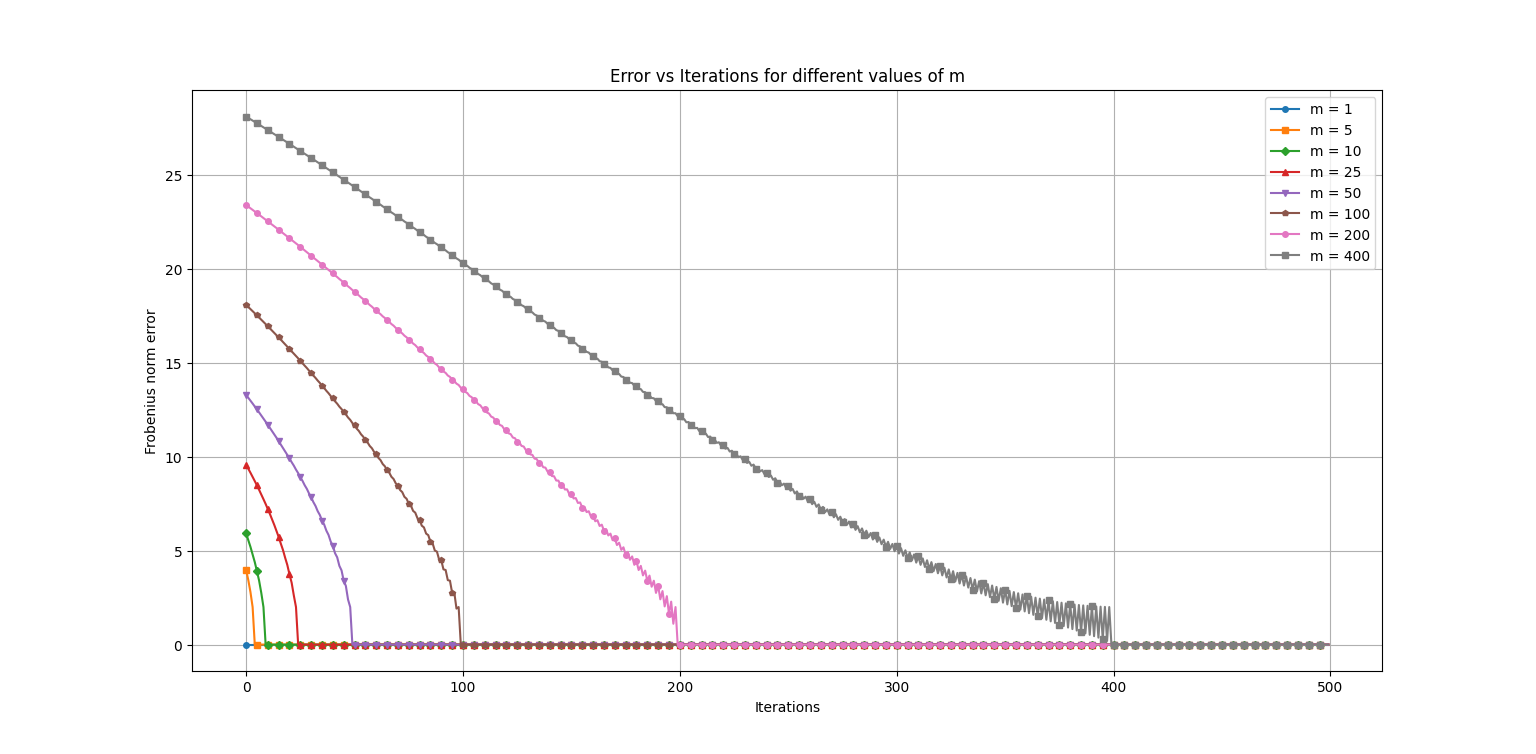}
        \caption{The first Householder vector is generated from an arbitrary Gaussian distribution, followed by normalization.
            The following Householder vectors are generated by retaining half of the entries from the previous Householder vector
            and generating the remaining entries from an arbitrary Gaussian distribution. The vector is then normalized.}
    \end{subfigure}
    \hfill
    \begin{subfigure}[t]{0.45\linewidth}
        \includegraphics[width=\linewidth]{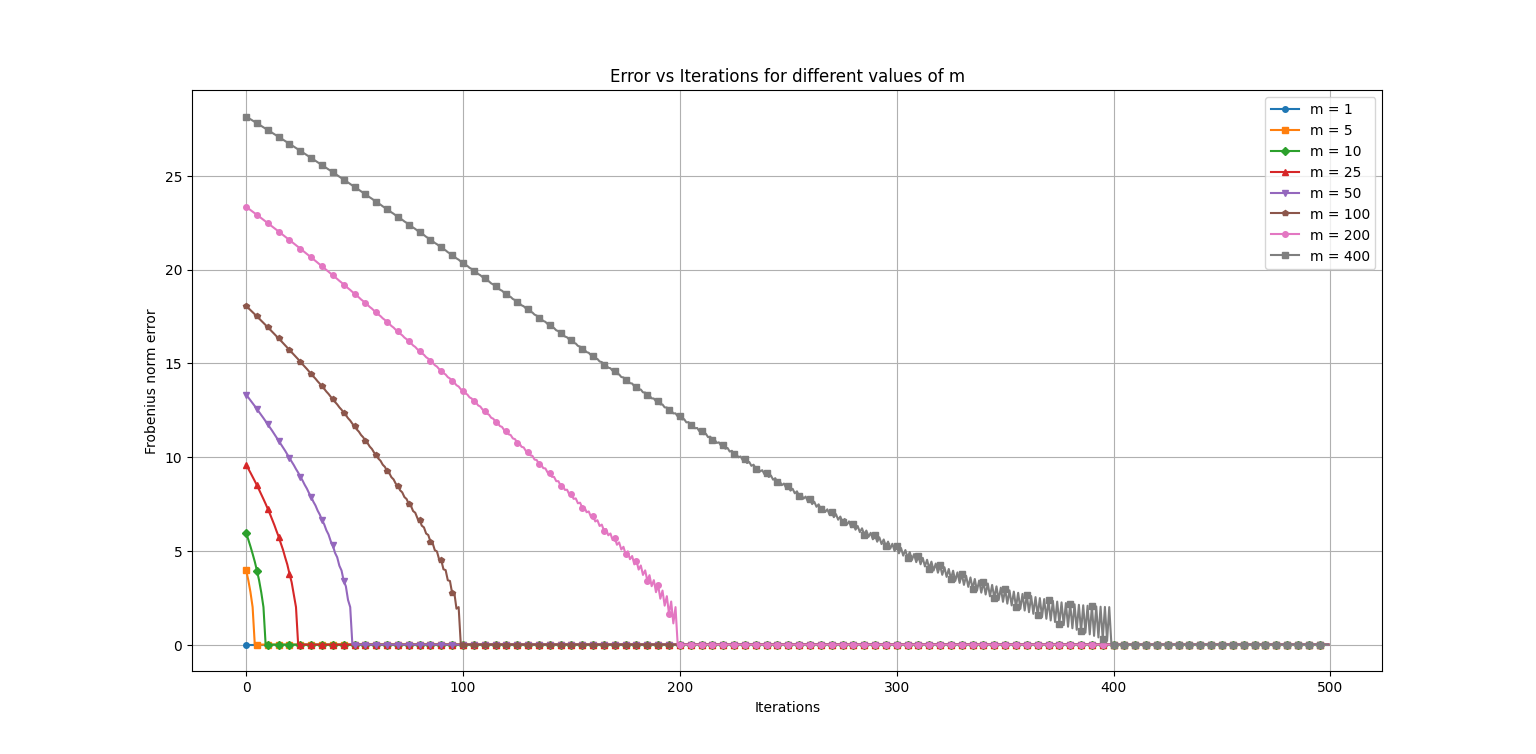}
        \caption{The Householder vectors are generated from an arbitrary Bernoulli distribution, with entries $1$ or $-1$ with some probability, followed by normalization.}
    \end{subfigure}

    \vspace{0.5cm}

    \begin{subfigure}[t]{0.45\linewidth}
        \includegraphics[width=\linewidth]{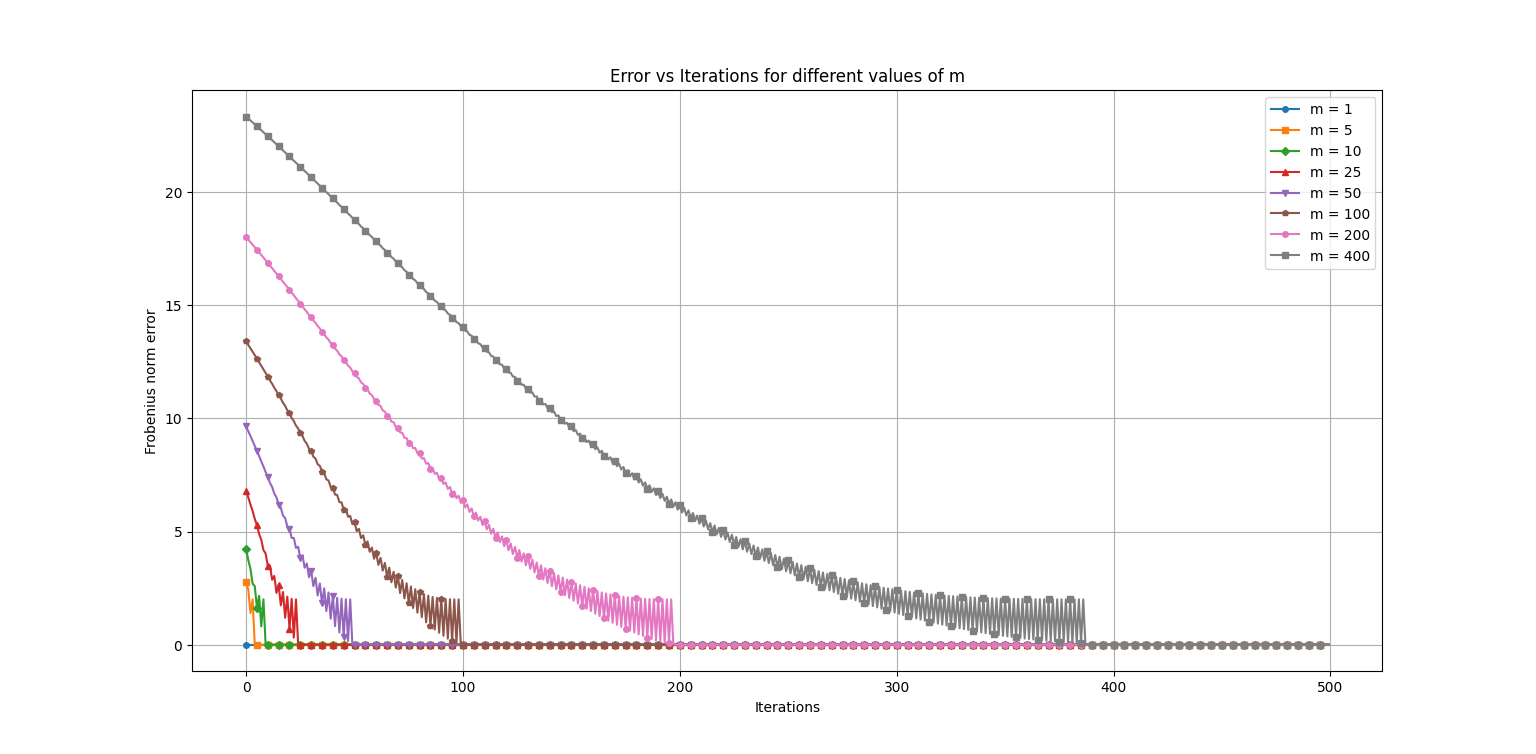}
        \caption{The Householder vectors are generated from an arbitrary Exponential distribution, followed by normalization.}
    \end{subfigure}
    \hfill
    \begin{subfigure}[t]{0.45\linewidth}
        \includegraphics[width=\linewidth]{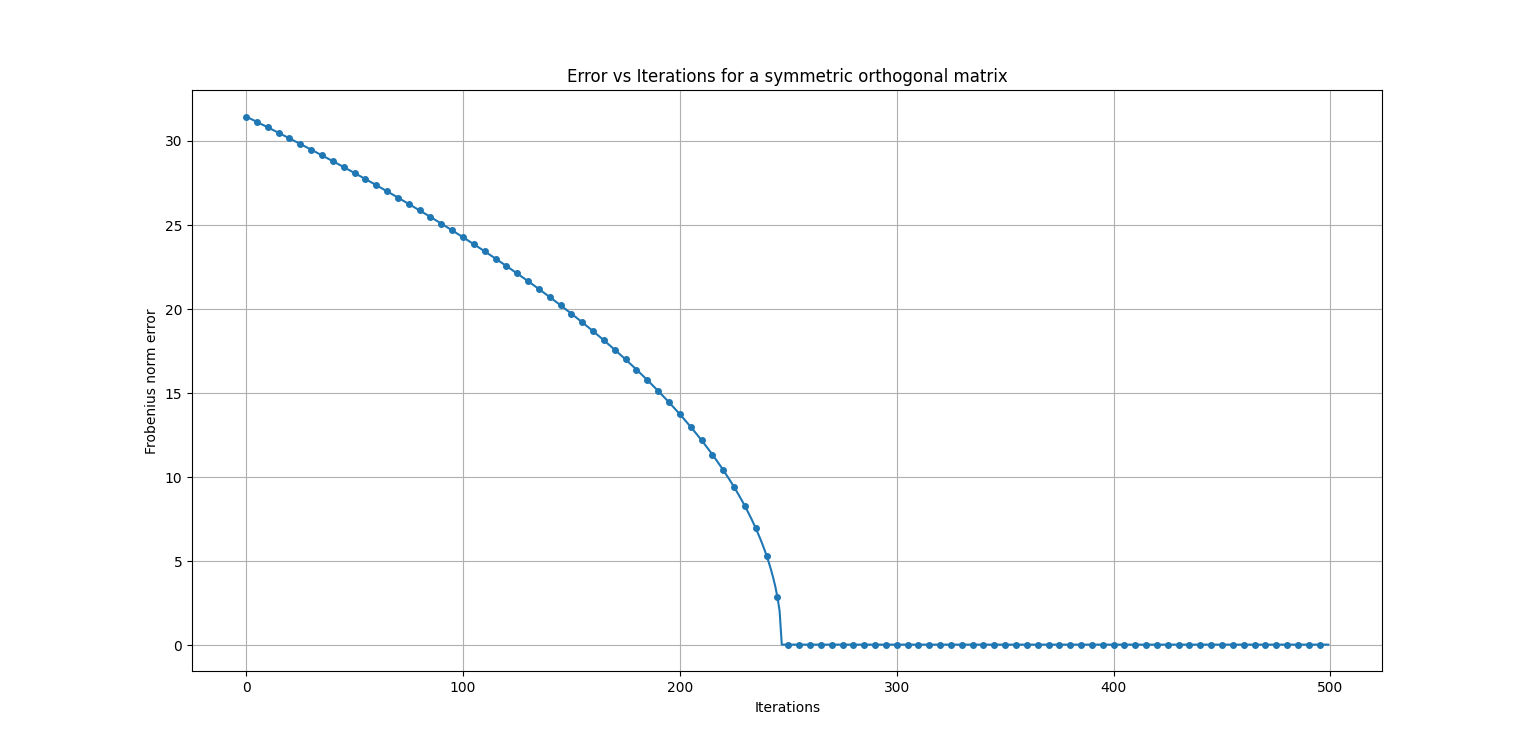}
        \caption{Error vs. iteration plot for symmetric orthogonal matrices.}
    \end{subfigure}

    \caption{Error vs. iteration plots for various distributions.}
    \label{fig:error_vs_iter}
\end{figure}

\begin{algorithm}[!ht]
    \small
    \caption{Approximating an orthogonal matrix in $\mathcal{H}_m$}
    \label{alg:ortho_prod_of_Householders}
    \textbf{Input:} $\mathbf{V},$ an orthogonal matrix, $\epsilon$, a small positive constant \\
    \textbf{Output:} $\mathbf{H}_1,\mathbf{H}_2,\ldots,\mathbf{H}_m$, $m$, the number of Householder matrices $\mathbf{V}$ is a product of
    \begin{algorithmic}[1]
        \State Set $k=0$
        \State $\hat{\mathbf{V}}=\mathbf{I}$
        \While{$\lVert \hat{\mathbf{V}} - \mathbf{V} \rVert_F$ $> \epsilon$ and $k<n$}
        \State Find the Eigen-decomposition of $\mathbf{(V}_k)_{\text{sym}}=(\mathbf{V}_k+\mathbf{V}_k^\textsf{T})/2$
        \State $\mathbf{H}_k \gets \mathbf{I}-2\mathbf{uu}^\textsf{T}$, where $\mathbf{u}$ is the eigenvector corresponding to the minimum eigenvalue $\lambda_{\min}$ of $\mathbf{(V}_k)_{\text{sym}}$.
        \State $\mathbf{V_{k+1}} \gets \mathbf{H}_k\mathbf{V}_k$
        \State $\hat{\mathbf{V}} \gets \hat{\mathbf{V}}\mathbf{H}_k$
        \State $k \gets k+1$
        \EndWhile
        \State Output: $\mathbf{H}_1,\mathbf{H}_2, \ldots \mathbf{H}_m$, $m$
    \end{algorithmic}
\end{algorithm}

\subsection{Simulations}

The orthogonal matrix $\mathbf{V}$ is chosen to be a $500 \times 500$ matrix generated as a
product of $m$ Householder matrices (for varying $m$ $[1,5,10,25,50,100,200,400]$). The Householder vectors corresponding to the
Householder matrices are generated from various distributions, sparsity levels, and relations between consecutive Householder matrices.
The figures displayed are for the aforementioned cases. The markers have been placed once every $5$ iterations. Note that
the error always goes down to $0$. Note, $\epsilon=0.05$ (refer Algorithm \ref{alg:ortho_prod_of_Householders}). If the original spectrum of $\mathbf{V}_{\text{sym}}$ has only a few positive eigenvalues, the onset of oscillations
is earlier, as is evident from the plot (e) (for the case where the Householder vectors were generated from an exponential distribution; see Figure \ref{fig:error_vs_iter}).

\section{Learning the Householder Dictionary with 2 samples}
\label{sec:3}
Given we can approximate $\mathbf{V} (\in \mathcal{H}_p)$ in $\mathcal{H}_m \ (m<p)$, it is imperative that we analyze the learnability of the Householder dictionary. We next discuss the main result for this problem.
\begin{theorem}  (\emph{Zero error achievability}) For the general model, $\mathbf{Y}=\mathbf{H}\mathbf{X}$, where $\mathbf{H}=\mathbf{I}-2\mathbf{uu}^\textsf{T}$ and $\mathbf{X}$ is an arbitrary binary matrix, $\mathbf{H}, \mathbf{X}$ can be \emph{uniquely} recovered with $p=2$ columns in $\mathbf{Y}$.
    \label{thm:zero_error}
\end{theorem}

The proof of Theorem \ref{thm:zero_error} proceeds by a brute force elimination of possibilities on the columns of $\mathbf{X}$: this is summarized in Algorithm \ref{find_HX_exponential}. We show that such an elimination uniquely identifies the vector $\mathbf{u}$ (Theorem \ref{thm:zero_error}) with $2$ columns. This implicitly also shows that any solution to the problem $\mathbf{Y}=\mathbf{H}\mathbf{X}$ is unique if it exists. This may not hold when the assumption on binary $\mathbf{X}$ is removed. Also note that while brute-force elimination succeeds, it does so with exponential time complexity.

\begin{theorem} (Non-uniqueness) If $\mathbf{X}$ is an arbitrary (non-binary) matrix, $\mathbf{H}, \mathbf{X}$ \emph{cannot} be \emph{uniquely} recovered (even up to permutation and sign) with any
    number of columns $p$. Thus, the assumption that $\mathbf{X}$ is a binary matrix is justified for recovery of the Householder dictionary.
    \label{thm:non-uniqueness}
\end{theorem}

\begin{algorithm}[!ht]
    \small
    \caption{Finding $\mathbf{H}, \mathbf{X}$ for $\mathbf{Y}=\mathbf{HX}$ with zero error}
    \label{find_HX_exponential}
    \textbf{Input:} $\mathbf{Y}$ \\
    \textbf{Output:} $\mathbf{H}, \mathbf{X}$
    \begin{algorithmic}[1]
        \While{all $n$ length binary vectors have not been exhausted}
        \State Set the first column of $\mathbf{X}$ as a random n length binary vector
        \State Find $\mathbf{u}$
        \EndWhile
        \State Repeat the above for the second column of $\mathbf{X}$
        \State If any of the $\mathbf{u}$ vectors obtained from a vector in the first column
        of $\mathbf{X}$ match with that of a $\mathbf{u}$ vector obtained from a random vector in the second column
        of $\mathbf{X}$, then set $\mathbf{H}$ as the corresponding Householder matrix, which is generated from $\mathbf{u}$
    \end{algorithmic}
\end{algorithm}

\subsection{Proof of Theorem \ref{thm:non-uniqueness}}

\begin{proof}
    We will show that there exists a pair of matrices,
    $(\mathbf{H}_1, \mathbf{X}_1)$ and another pair $(\mathbf{H}_2, \mathbf{X}_2)$, such that
    $\mathbf{H}_1\mathbf{X}_1 = \mathbf{H}_2\mathbf{X}_2$. Consider the householder vectors: \(\mathbf{u}_1^\textsf{T} = \left(\sqrt{{1}/{3}}, \sqrt{{2}/{3}}\right)\), \(\mathbf{u}_2^\textsf{T} = \left({1}/{\sqrt{2}}, {1}/{\sqrt{2}}\right)\). The corresponding householder matrices are \(\mathbf{H}_1 = (\mathbf{v}_1\ \mathbf{v}_2)\), where \(\mathbf{v}_1 = ({1}/{3}, {-2\sqrt{2}}/{3})^\textsf{T}\) and \(\mathbf{v}_2 = ({-2\sqrt{2}}/{3}, {-1}/{3})^\textsf{T}\); and \(\mathbf{H_2} = (\mathbf{w}_1\ \mathbf{w}_2)\), where \(\mathbf{w}_1 = (0, -1)^\textsf{T}\) and \(\mathbf{w}_2= (-1, 0)^\textsf{T}\)
    Consider the \(p^{th}\) column vector of \(\mathbf{X}_1\) and \(\mathbf{X}_2\) as \(\mathbf{X}_{1p}=(x_{11p}, x_{12p})^\textsf{T}\) and \(\mathbf{X}_{2p}=(x_{21p}, x_{22p})^\textsf{T}\). Thus, the corresponding column vectors of the \(\mathbf{Y}\) matrices are \(\mathbf{Y}_{1p} = \left(({x_{11p}-2\sqrt{2}x_{12p}})/{3}, {(-x_{12p}-2\sqrt{2}x_{11p}})/{3}\right)^\textsf{T}\) and \(\mathbf{Y}_{2p} = (-x_{22p}, -x_{21p})^\textsf{T}\). We need \(\mathbf{Y}_{1p} = \mathbf{Y}_{2p}\), thus we need \(({x_{11p}-2\sqrt{2}x_{12p}})/{3} = -x_{22p}\) and \({(-x_{12p}-2\sqrt{2}x_{11p}})/{3} = -x_{21p}\). This can be done for every column \(p\). For example, a satisfying assignment would be \(\mathbf{X}_{1p} = \left({2\sqrt2}/{3}, {1}/{3}\right)^\textsf{T}\) and \(\mathbf{X}_{2p} = (1, 0)^\textsf{T}\).

    Thus, the solution is not unique to sign and permutation. Note that even for
    Theorem 3, multiple $\mathbf{H}, \mathbf{X}$ pairs could give the
    same $\mathbf{Y}$ for a given column. However, we rely on the fact that entries of
    the column vectors of $\mathbf{X}$ are constrained to be 0 or 1 only.
    This would ensure that up to permutation, only $O(n)$ such possibilities existed.
    However, without such a restriction here, infinite solutions are possible and thus, recovery
    is impossible.

\end{proof}

\subsection{Proof of Theorem \ref{thm:zero_error}}

\begin{proof}
    We try all possible combinations of binary vectors for the columns of $\mathbf{X}$- thus, we effectively check all possible cases.
    If the equations are consistent, the $\mathbf{u}$ vector obtained from the first column of $\mathbf{X}$ will match with the $\mathbf{u}$ vector obtained from the second column of $\mathbf{X}$.
    To identify the $\mathbf{u}$ vector, we need to solve the n simultaneous equations
    \(   \mathbf{Y}_{ij} = \sum_{k=1}^{n} \mathbf{H}_{ik}\mathbf{X}_{kj}\) resulting in
    \begin{equation}
        \mathbf{Y}_{ij} = \sum_{k=1}^{n} \left(\delta_{ik}-2u_iu_k\right)\mathbf{X}_{kj} \quad \forall i \in [n] \label{2*}
    \end{equation}

    Here $\delta_{ik}$ is the standard indicator for $j=k$. For any "guess" of the first column of $\mathbf{X}$, we can solve the above $n$ equations to get the $\mathbf{u}$ vector.
    We will get $2^n$ such $\mathbf{u}$ vectors for the first and the second column. If any of the $\mathbf{u}$ vectors obtained
    from the first column of $\mathbf{X}$ match with that of a $\mathbf{u}$ vector obtained from the second column of
    $\mathbf{X}$, then we have found the correct $\mathbf{u}$ vector. The goal is to show that there will be at most one vector $\mathbf{u}$ that works for both columns of $\mathbf{X}$. We prove this by contradiction. Assume that $(\mathbf{X}_1, \mathbf{u}_1)$ and $(\mathbf{X}_2, \mathbf{u}_2)$ satisfy \eqref{2*} (let the corresponding householder matrices be $\mathbf{H}_1$ and $\mathbf{H}_2$ respectively).
    Define the sets (of non-supports) $S_{lj}=:\{k: \mathbf{X}_{lkj}=0\}$ for $j,l \in \{1,2\}$: this is the location of zeroes in the $j^{th}$ column of the $l^{th}$ solution.
    From \eqref{2*}, we have \(\sum_{ k \notin S_{1j}}(\delta_{ik}-2u_{1i}u_{1k}) = \sum_{k \notin S_{2j}} (\delta_{ik}-2u_{2i}u_{2k})  \).  \\

    Define \(c_{lj} = \sum_{i \in S_{lj}} u_{li} \) for $l,j \in \{1,2\}$ as the sum of the entries in the $l^{th}$ solution for $\mathbf{u}$ at locations where the corresponding coefficient matrix is zero in the $j^{th}$ column. Likewise, let $c_{l} = \sum u_{li}$ be the sum of all entries in the $l^{th}$ solution for $\mathbf{u}$. (Thus, define $c_{s1}=c_{1j}$ and $c_{s2}=c_{2j}$). For notational simplicity, let $c_1-c_{s1}=\delta_1$ and $c_2-c_{s2}=\delta_2$.
    Now writing \eqref{2*} for $i$ in each of the sets $T\cap S$ where $T \in \{S_{1j}, S_{1j}^c \}$ and $S \in \{S_{2j}, S_{2j}^c \}$. Define sets $P_1,P_2,P_3,P_4$ corresponding to the four cases: i \( \notin S_{1j},\  i \notin S_{2j}; \ i \notin S_{1j}, \ i \in S_{2j}; \ i \in S_{1j}, \ i \in S_{2j}; \  i \in S_{1j}, \ i \notin S_{2j}    \).

    \begin{Lemma}
        Following the notation from the previous paragraph, in case there are two solutions to \eqref{2*}, we can obtain the second solution from the first as
        \(
        u_{2i}  = (\delta_1/\delta_2) u_{1i} \ i \in P_1                                         ; \
        u_{2i}  = ({\delta_1}/{\delta_2})u_{1i} - \left({1}/{2 \delta_2}\right) \ i \in P_2               ; \
        u_{2i}  = ({\delta_1}/{\delta_2})u_{1i} \ i \in P_3                                               ; \
        u_{2i}  = ({\delta_1}/{\delta_2)}u_{1i} + \left({1}/{2 \delta_2}\right)  i \in P_4
        \)
        \label{lem:dict_equal}
    \end{Lemma}

    \begin{proof}
        We show that the $u_{2i}$'s can be expressed in terms of the $u_{1i}$'s for each case. Here, $i \in [n]$ indicates the $i^{th}$ entry of the $\mathbf{u}$ vector,      corresponding to the $i^{th}$ row of $\mathbf{X}$. This leads to the following equality constraints:
        \begin{align*}
            \sum_{i \notin S_{1j}} \left(\delta_{ik}-2u_{1i}u_{1k}\right) & = \sum_{i \notin S_{2j}} \left(\delta_{ik}-2u_{2i}u_{2k}\right) \\
            \sum_{i \notin S_{1j}} \left(\delta_{ik}-2u_{1i}u_{1k}\right) & = \sum_{i \in S_{2j}} \left(\delta_{ik}-2u_{2i}u_{2k}\right)    \\
            \sum_{i \in S_{1j}} \left(\delta_{ik}-2u_{1i}u_{1k}\right)    & = \sum_{i \in S_{2j}} \left(\delta_{ik}-2u_{2i}u_{2k}\right)    \\
            \sum_{i \in S_{1j}} \left(\delta_{ik}-2u_{1i}u_{1k}\right)    & = \sum_{i \notin S_{2j}} \left(\delta_{ik}-2u_{2i}u_{2k}\right)
        \end{align*}
        (Note that the summations are over $k$). These can be simplified as:
        \(
        1-2u_{1i}\delta_1 = 1-2u_{2i}\delta_2 \ \ \ i \in P_1 ; \ \
        1-2u_{1i}\delta_1 = -2u_{2i}\delta_2 \ \ \ i \in P_2 ;    \ \
        -2u_{1i}\delta_1  = -2u_{2i}\delta_2 \ \ \ i \in P_3 ;       \ \
        -2u_{1i}\delta_1  = 1-2u_{2i}\delta_2 \ \ \ i \in P_4
        \)
        On rearranging the terms, we get the required result. We assume that we do not divide by 0 in any case.
    \end{proof}

    Now, consider the following Lemma:
    \begin{Lemma}
        Following the notation above,
        \begin{enumerate}
            \item $\lvert P_2 \rvert = \lvert P_4 \rvert$,
            \item        \( (\delta_2)/(\delta_1) = c_1/c_2, \)
            \item \(u_{1i}/u_{2i} = c_1/c_2 \) for  $i\in P_1\cup P_3,$
            \item \(\sum_{i \in P_2 \cup P_4} u_{1i}/\sum_{i \in P_2 \cup P_4} u_{2i}  = c_1/c_2. \)
        \end{enumerate}
        \label{lem:3}
    \end{Lemma}

    \begin{proof}
        First, we use the fact that any solution to $\mathbf{u}$ must be unit norm.
        \begin{equation}
            \begin{aligned}
                \sum_{i=1}^{n} u_{2i}^2 & = 1 \label{eq:2} \\
            \end{aligned}
        \end{equation}
        Substituting the equations from Lemma \ref{lem:dict_equal} into \eqref{eq:2}, we get
        \(
        \sum_{P_1}
        \left( (\delta_1 / \delta_2) u_{1i} \right)^2 +
        \sum_{P_2}
        \left( (\delta_1 / \delta_2) u_{1i} -  (1 / 2 \delta_2) \right)^2
        + \sum_{P_3}
        \left( (\delta_1 / \delta_2) u_{1i} \right)^2 +
        \sum_{P_4}
        \left( (\delta_1 / \delta_2) u_{1i} +  (1 / 2 \delta_2) \right)^2 = 1
        \). By grouping the terms appropriately, we get \\
        \(
        (\delta_1 / \delta_2)^2
        \left( \sum_{\cup_{i=1}^{4} P_i} u_{1i}^2 \right)
        + \sum_{P_2 \cup P_4}  (1 / 2 \delta_2)^2
        - \sum_{P_2} (\delta_1 / \delta_2^2) u_{1i}
        + \sum_{P_4} (\delta_1 / \delta_2^2) u_{1i} = 1
        \). On using the unit norm condition on $\mathbf{u_1}$:
        \(
        (\delta_1 / \delta_2)^2
        + \sum_{P_2 \cup P_4} (1 / 2 \delta_2)^2
        - \sum_{P_2} (\delta_1 / \delta_2^2) u_{1i}
        + \sum_{P_4} (\delta_1 / \delta_2^2) u_{1i} = 1
        \). On cross multiplying and expanding:
        \(
        \delta_1^2
        + {(\lvert P_2 \rvert + \lvert P_4 \rvert)}/{4}
        - \sum_{P_2} \delta_1 u_{1i}  + \sum_{P_4} \delta_1 u_{1i}
        = \delta_2^2
        \). Which, on simplifying, leads to:
        \begin{equation}
            \begin{aligned}
                \delta_2^2 - \delta_1^2
                = {(\lvert P_2 \rvert + \lvert P_4 \rvert)}/{4}
                - \sum_{P_2} \left( \delta_1 \right) u_{1i}
                + \sum_{P_4} \left( \delta_1 \right) u_{1i}.
            \end{aligned}
            \label{eq:3}
        \end{equation}

        The equations derived above serve as the basis for the proof of the Lemma. We now proceed to prove the individual parts of the Lemma.

        \begin{enumerate}
            \item
                  Consider the equation obtained from the sums of entries of the
                  $\mathbf{u}$ vectors.
                  \begin{align*}
                      \sum_{i=1}^{n} u_{2i} & = c_2
                  \end{align*}
                  Plugging in the expressions from Lemma \ref{lem:dict_equal}, we get
                  \(
                  \sum_{P_1} \left( \delta_1 / \delta_2 u_{1i} \right)
                  + \sum_{P_2} \left( \delta_1 / \delta_2 u_{1i} - \left( 1 / (2 \delta_2) \right) \right)
                  + \sum_{P_3} \left( \delta_1 / \delta_2 u_{1i} \right)
                  + \sum_{P_4} \left( \delta_1 / \delta_2 u_{1i} + \left( 1 / (2 \delta_2) \right) \right) = c_2.
                  \) Thus,
                  \(
                  (\delta_1 / \delta_2) \left( \sum_{\cup_{i=1}^{4} P_i} u_{1i} \right)
                  - \sum_{P_2} (1 / 2 \delta_2) + \sum_{P_4}  (1 / 2 \delta_2) = c_2
                  \). This gives
                  \(
                  c_1 \left( \delta_1 / \delta_2 \right)
                  - \sum_{P_2}  (1 / 2 \delta_2) + \sum_{P_4} (1 / 2 \delta_2) = c_2
                  \). On cross-multiplying, we get
                  \(
                  c_1\left(\delta_1 \right)
                  - \sum_{P_2}{1}/{2} + \sum_{P_4}{1}/{2} = c_2\delta_2
                  \). On simplifying, we have,
                  \begin{equation}
                      \begin{aligned}
                          - \sum_{P_2}{1}/{2} +
                          \sum_{P_4}{1}/{2} & = c_2\delta_2 - c_1 \delta_1 \label{eq:4}
                      \end{aligned}
                  \end{equation}
                  We then consider sums over $P2$ and $P3$:
                  \(
                  \sum_{P2 \cup P3} u_{2i} = c_{s2}
                  \). This gives
                  \(
                  \sum_{P2}u_{2i} + \sum_{P3}u_{2i} = c_{s2}
                  \). Substituting the expressions from Lemma \ref{lem:dict_equal}, we get
                  \(
                  \sum_{P2 \cup P3} \left({\delta_1}/{\delta_2}\right)u_{1i}
                  - \sum_{P_2}\left({1}/{2 \delta_2}\right) = c_{s2}
                  \). On cross-multiplying and rearranging,
                  \begin{equation}
                      \begin{aligned}
                          \sum_{P2 \cup P3} \delta_1u_{1i}
                          - \sum_{P_2}{1}/{2} & = c_{s2}\delta_2 \label{eq:5}
                      \end{aligned}
                  \end{equation}
                  Correspondingly, consider the sums over $P3$ and $P4$:
                  \(
                  \sum_{P3 \cup P4} u_{1i} = c_{s1}
                  \). Decomposing similar to the previous case,
                  \(
                  \sum_{P2}u_{1i} + \sum_{P3}u_{1i} = c_{s1}
                  \). Substituting the expressions from Lemma \ref{lem:dict_equal}, \
                  \(
                  \sum_{P3 \cup P4} \left({\delta_2}/{\delta_1}\right)u_{2i}
                  - \sum_{P_4}\left({1}/{2 \delta_1}\right) = c_{s1}
                  \). Simplifying and rearranging,
                  \begin{equation}
                      \begin{aligned}
                          \sum_{P3 \cup P4} \left(\delta_2\right)u_{2i}
                          - \sum_{P_4}{1}/{2} & = c_{s1}\delta_1 \label{eq:6}
                      \end{aligned}
                  \end{equation}

                  Thus, we get (from Equations \ref{eq:5} and \ref{eq:6}):
                  \begin{equation}
                      \begin{aligned}
                          c_{s2}\delta_2 - c_{s1}\delta_1 & =  \sum_{P2} \delta_1u_{1i}
                          -  \sum_{P4} \delta_2u_{2i} - \sum_{P_2}{1}/{2} +
                          \sum_{P_4}{1}/{2} \label{eq:7}
                      \end{aligned}
                  \end{equation}

                  Furthermore (from Equations \ref{eq:4} and \ref{eq:7}):
                  \(
                  c_2\delta_2 - c_1\delta_1-(c_{s2}\delta_2 - c_{s1}\delta_1)
                  = - \left(\sum_{P2} \delta_1u_{1i} -  \sum_{P4} \delta_2u_{2i}\right)
                  \). This gives us:
                  \begin{equation}
                      \begin{aligned}
                          \delta_2^2-\delta_1^2 & = - \sum_{P2} \delta_1u_{1i} +  \sum_{P4} \delta_2u_{2i} \label{eq:8}
                      \end{aligned}
                  \end{equation}

                  Using the above and Equation $\ref{eq:3}$, we get:
                  \(
                  {(\lvert P_2 \rvert + \lvert P_4 \rvert)}/{4} = \sum_{P4} \delta_2u_{2i} - \sum_{P4} \delta_1u_{1i}
                  \). This gives us:
                  \(
                  \lvert P_2 \rvert + \lvert P_4 \rvert = 2\sum_{P4} -2u_{1i}\delta_1
                  + 2u_{2i}\delta_2 \). Or, \(
                  \lvert P_2 \rvert + \lvert P_4 \rvert = 2\sum_{P4} 1-2u_{2i}\delta_2+2u_{2i}\delta_2 \). Hence, \(
                  \lvert P_2 \rvert + \lvert P_4 \rvert = 2\sum_{P4} 1
                  \)
                  Thus, we have
                  \begin{equation}
                      \begin{aligned}
                          \lvert P_2 \rvert - \lvert P_4 \rvert & =  0 \label{eq:9}
                      \end{aligned}
                  \end{equation}

                  This concludes the proof of the first part of the Lemma. Note that this could have been obtained directly
                  by taking the $l_2$ norm (column-wise) on both sides of the equation $\mathbf{H}_1\mathbf{X}_1=\mathbf{H}_2\mathbf{X}_2$. However, the following results
                  are not as straightforward to obtain without the analysis of the underlying structure of the problem.

            \item
                  Using Equations $\ref{eq:9}$ and $\ref{eq:4}$, we have:
                  \(
                  0 = c_2\delta_2 - c_1\left(\delta_1\right)
                  \)
                  On rearranging the terms, we get
                  \begin{equation}
                      \begin{aligned}
                          {c_1}/{c_2} & = {\delta_2}/{\delta_1} \label{eq:10}
                      \end{aligned}
                  \end{equation}

            \item
                  Using Lemma $\ref{lem:dict_equal}$, and substituting the result 2. from Lemma $\ref{lem:3}$, we get:
                  \begin{equation}
                      \begin{aligned}
                          {u_{1i}}/{u_{2i}} & = {c_1}/{c_2} \ \ \ \text{for}  \ \ \ P_1 \cup P_3 \label{eq:13}
                      \end{aligned}
                  \end{equation}

            \item
                  Using Equations $\ref{eq:7} $ and $ \ref{eq:9}$, we get:
                  \(
                  c_{s2}\delta_2 - c_{s1}\delta_1 =  \sum_{P2} \delta_1u_{1i} - \sum_{P4} \delta_2u_{2i}
                  \). Dividing the equation by $\delta_1$, we get
                  \(
                  c_{s2}({\delta_2}/{\delta_1})-c_{s1} =  \sum_{P2} u_{1i} - \sum_{P4} ({\delta_2}/{\delta_1})u_{2i}
                  \)
                  Using result 2. from Lemma $\ref{lem:3}$, we get:
                  \(
                  c_{s2}({c_1}/{c_2})-c_{s1} = \sum_{P2} u_{1i} - \sum_{P4} ({c_1}/{c_2})u_{2i}
                  \)
                  Thus, we have:
                  \(
                  ({c_1}/{c_2})\sum_{P2 \cup P3 \cup P4} u_{2i} = \sum_{P2 \cup P3 \cup P4} u_{1i}
                  \)
                  Or equivalently,
                  \begin{equation}
                      \begin{aligned}
                          ({c_1}/{c_2}) & = \frac{\sum_{P2 \cup P3 \cup P4} u_{1i}}{\sum_{P2 \cup P3 \cup P4} u_{2i}} \label{eq:11}
                      \end{aligned}
                  \end{equation}

                  Using Equations $\ref{eq:11} $ and $ \ref{eq:13}$:
                  \(
                  \sum_{P2 \cup P3 \cup P4} u_{1i} = \sum_{P2 \cup P3 \cup P4} u_{2i} ({c_1}/{c_2}) \\
                  \)
                  Separating the terms, we get:
                  \(
                  \sum_{P2 \cup P4} u_{1i} + \sum_{P3} u_{1i} = \sum_{P2 \cup P4} u_{2i}({c_1}/{c_2}) + \sum_{P3} u_{2i} ({c_1}/{c_2}) \\
                  \)
                  Using result 3. from Lemma $\ref{lem:3}$, we get:
                  \(
                  \sum_{P2 \cup P4} u_{1i} = \sum_{P2 \cup P4} u_{2i} ({c_1}/{c_2})
                  \)
                  On simplifying, we have the result:
                  \begin{equation}
                      \begin{aligned}
                          \frac{\sum_{P2 \cup P4} u_{1i}}{\sum_{P2 \cup P4} u_{2i}} & = ({c_1}/{c_2})
                      \end{aligned}
                  \end{equation}

        \end{enumerate}
    \end{proof}

    According to result $1.$ in Lemma \ref{lem:3}, the two matrices are only different due to permutation. Now, we consider multiple columns. We know that every
    column of $\mathbf{X}$ has a permutation of 0's and 1's corresponding to the
    ground truth $\mathbf{u}$ vector. The corresponding column vector in
    $\mathbf{Y}$ is different for different columns. However, the ground truth $\mathbf{u}$ vector
    remains the same throughout. Thus, the $\mathbf{u}$ vector generated from one of the permutations
    will certainly match that generated in the previous columns. An error is caused
    when the same "incorrect" $\mathbf{u}$ vector is generated for every column.
    Consider Lemma \ref{lem:dict_equal}. Say we know the ground truth $\mathbf{u}$ vector, $\mathbf{u_1}$. Using the condition on unit norm of $\mathbf{u}$, we also have: \((\delta_2) = \pm ((\delta_1)^2+(\delta_1)\left(\sum_{P_4}u_{1i} - \sum_{P_2}u_{1i}\right)+({\lvert P_2 \rvert + \lvert P_4 \rvert})/{4})^{1/2}.\)
    For this value to always be defined, \((\delta_1)^2+(\delta_1)\left(\sum_{P_4}u_{1i} - \sum_{P_2}u_{1i}\right)+({\lvert P_2 \rvert + \lvert P_4 \rvert})/{4} \geq 0 \ \forall \mathbf{u_1}\). \\
    Thus, the discriminant of the above quadratic must always be non-positive. Consequently, we have, \(\Delta  = \left(\sum_{P_4}u_{1i} - \sum_{P_2}u_{1i}\right)^2 - 4\left({\lvert P_2 \rvert + \lvert P_4 \rvert}/{4}\right) \). Since \(\max \left(\sum_{P_4}u_{1i} - \sum_{P_2}u_{1i}\right)  = ({\lvert P_2 \rvert + \lvert P_4 \rvert})^{1/2}\), we have that \(\Delta \leq 0\) for all \(\mathbf{u_1}.\)

    Thus, when given a $\mathbf{u_1}$ vector, the corresponding possible $\mathbf{u_2}$ vector can be found.
    Only 2 of these can exist for a given ($\mathbf{u_1}$ vector, column vector of $\mathbf{X_1}$) pair, and only one can exist up to sign.
    Now, consider the
    second column. We need to check if the same $\mathbf{u_2}$ vector obtained from the first column can satisfy the equations
    corresponding to the second column. Note that there is a conflict
    iff a $\mathbf{u_2}$ vector generated from the second column is
    exactly the same as that generated from the first column. This is because,
    when we solve the equations to find the ground truth $\mathbf{u_1}$ vector, the result is exactly
    the same from both columns, since this is how the matrix $\mathbf{Y}$ was generated.
    There is no variation even in permutation or sign. Thus, we check if this is possible. Assume that the
    $\mathbf{u_2}$ vector generated from the second column is the same as that generated from the first column. This implies that both
    $c_1$ and $c_2$ are the same.

    Consider the equations from Lemma $\ref{lem:dict_equal}$.
    Using Lemma $\ref{lem:3}$, these are equal to:
    \(
    u_{1i}  = u_{2i}c'\text{ for }i \in P_1 \cup P_3\), \( u_{1i} = c''+u_{2i}c' \text{ for } i \in P_2 \) and \(u_{1i} = -c''+u_{2i}c' \text{ for } i \in P_4\) where $c'=c_1/c_2$, $c''= 1/(2\delta_1)$.
    Note that since $\mathbf{u_1}$ is the same, $c_1$ is the same in the equations above as
    that from the first column. Thus, if $\mathbf{u_2}$ has to be the same
    , then $\delta_1$ must be the same for the ground $\mathbf{u_1}$ vector generated from the first column.
    Thus, $c_{s1}$ must also be the same. Not only this, but $\lvert P_2 \rvert$ and
    $\lvert P_4 \rvert$ must be the same because there is a one-to-one relationship
    between the known ground truth $\mathbf{u_1}$ vector and the $\mathbf{u_2}$ vector generated from any column.
    If $\lvert P_2 \rvert $ is different, the number of entries that are calculated using the above equation
    for $i \in P_2$ will change. Thus, the $\mathbf{u_2}$ vector generated will be different. The only way it won't change is
    if ${1}/\delta_1$ is $0$. But this is not possible.
    Therefore, the only way of getting an identical $\mathbf{u}_2$ vector is if the ground truth $\mathbf{X}$ column vector for the
    first and second columns are identical. But this is a contradiction since we assume them to be different. Thus, we will recover the ground truth $\mathbf{u}_1$ vector uniquely by using only 2 columns.
\end{proof}

\begin{table*}[h!]
    \centering
    \scriptsize
    \renewcommand{\arraystretch}{1.5}
    \begin{tabular}{|c|c|c|c|c|}
        \hline
        Possible $\mathbf{x}$  & Column 1: $[u_1,u_2,u_3]$                        & Valid $\left(\lVert u \rVert =1 \right)$?                 & Column 2: $[u_1,u_2,u_3]$                               & Valid $\left(\lVert u \rVert =1 \right)$? \\
        \hline
        $[0, 0, 0]^\textsf{T}$ & No solution                                      & No                                                        & No solution                                             & No                                        \\
        \hline
        $[0, 0, 1]^\textsf{T}$ & \([1/\sqrt{42}, -1/\sqrt{42}, \sqrt{7/6}]\)
                               & No                                               & \fcolorbox{red}{white}{\([2/3, 1/3, 2/3]\)}               & Yes                                                                                                 \\
        \hline
        $[0, 1, 0]^\textsf{T}$ & \([1/(2\sqrt{3}), 1/\sqrt{3}, 2/\sqrt{3}]\)
                               & No                                               & \([4\sqrt{2/13}/3, \sqrt{13/2}/3, 1/(3\sqrt{26})]\)       & No                                                                                                  \\
        \hline
        $[0, 1, 1]^\textsf{T}$ & \([1/(3\sqrt{6}), \sqrt{2/3}/3, 7/(3\sqrt{6})]\)
                               & Yes                                              & \([(4\sqrt{2/21})/3, 13/(3\sqrt{42}), (4\sqrt{2/21})/3]\) & No                                                                                                  \\
        \hline
        $[1, 0, 0]^\textsf{T}$ & \([\sqrt{2/3}, -1/(2\sqrt{6}), \sqrt{2/3}]\)
                               & No                                               & \([(\sqrt{17/2})/3, (2\sqrt{2/17})/3, 1/(3\sqrt{34})]\)
                               & No                                                                                                                                                                                                                 \\
        \hline
        $[1, 0, 1]^\textsf{T}$ & \([2\sqrt{2/33}, -1/\sqrt{66}, 7/\sqrt{66}]\)
                               & Yes                                              & \([17/(15\sqrt{2}), (2\sqrt{2})/15, (4\sqrt{2})/15]\)
                               & No                                                                                                                                                                                                                 \\
        \hline
        $[1, 1, 0]^\textsf{T}$ & \fcolorbox{red}{white}{\([2/3, 1/3, 2/3]\)}      & Yes                                                       & \([17/(6\sqrt{15}), 13/(6\sqrt{15}), -1/(6\sqrt{15})]\)
                               & No                                                                                                                                                                                                                 \\
        \hline
        $[1, 1, 1]^\textsf{T}$ & \([2\sqrt{2/39}, \sqrt{2/39}, 7/\sqrt{78}]\)
                               & No                                               & \([17/(6\sqrt{19}), 13/(6\sqrt{19}), 4/(3\sqrt{19})]\)
                               & No                                                                                                                                                                                                                 \\
        \hline
    \end{tabular}
    \caption{}
    \label{tab:illustration}
\end{table*}

\subsection{Illustration of Theorem \ref{thm:zero_error}}
Consider an arbitrary $3 \times 3$ householder matrix. Say, the householder vector $\mathbf{u}$ is given by: \( \mathbf{u} = [2/3, 1/3, 2/3]^\textsf{T}\). $\mathbf{H}=\mathbf{I}-2\mathbf{uu}^\textsf{T}$. Thus, the corresponding data matrix, householder matrix, and arbitrarily chosen binary matrix $\mathbf{X}$ are as follows:
\[
    \underbrace{
        \begin{bmatrix}
            -1/3 & -8/9 & \cdots \\
            1/3  & -4/9 & \cdots \\
            -4/3 & 1/9  & \cdots
        \end{bmatrix}
    }_{\mathbf{Y}}
    =
    \underbrace{
        \begin{bmatrix}
            1/9  & -4/9 & -8/9 \\
            -4/9 & 7/9  & -4/9 \\
            -8/9 & -4/9 & 1/9
        \end{bmatrix}
    }_{\mathbf{H}}
    \underbrace{
        \begin{bmatrix}
            1 & 0 & \cdots \\
            1 & 0 & \cdots \\
            0 & 1 & \cdots
        \end{bmatrix}
    }_{\mathbf{X}}
\]

On applying the algorithm to $\mathbf{Y}$, we get the results from Table \ref{tab:illustration}. That is, we express the unknown householder matrix in its general form:
\[
    \mathbf{H} = \begin{bmatrix}
        1 - 2u_1^2 & -2u_1u_2   & -2u_1u_3   \\
        -2u_1u_2   & 1 - 2u_2^2 & -2u_2u_3   \\
        -2u_1u_3   & -2u_2u_3   & 1 - 2u_3^2
    \end{bmatrix}
\]
Using this, we solve the equations obtained for all cases. The only common solution has been highlighted in the table (see Table \ref{tab:illustration}). Note originally, we don't have either $\mathbf{H}$ or $\mathbf{X}$. The algorithm gives us $\mathbf{H}$, and $\mathbf{X}$ can be found as the column of $\mathbf{X}$ used to obtain the ground truth $\mathbf{H}$.

\section{Discussion and Future Work}
We note that approximations on $\mathcal{H}_m$ may require
a large value of $m$ even for matrices like $-\mathbf{I}$:
this is because the individual Householder matrices, which
are used as building blocks have a large eigenspace
corresponding to the eigenvalue $1$. We see that $\mathbf{-I}$
can be represented as a product of $n$ Householder matrices
as in \eqref{eq:symmetric_Householder_decomposition}. In
fact, there doesn't exist a representation of $\mathbf{-I}$
as a product of a $m<n$ Householder matrices
(this can be seen as a consequence of Theorem \ref{thm:correctness}).
Since our primary interest is in computational and storage efficiency,
we may consider modifying the basic Householder unit used as building blocks.
A possible consoderation is using the following modified Householder matrices $z_1\mathbf{I}-z_2\mathbf{uu}^T $ (where $z_1, z_2$ are tunable constants) as fundamental blocks.
his will hopefully help us express a richer class of orthogonal matrices using a smaller number of fundamental blocks.

The main takeaway from Section \ref{sec:3} is that meaningful lower bounds on sample complexity for structured dictionary learning can only be obtained by restricting the dictionary learning algorithm to be polynomial time.

\end{document}